\documentclass[12pt]{article}
 \usepackage{a4wide}
\usepackage{authblk}

\usepackage[breaklinks]{hyperref}
\usepackage{booktabs}

\usepackage{mathptmx}       
\usepackage{graphicx}        

\usepackage[utf8]{inputenc}
\usepackage{amsmath,amssymb,amsthm}
\usepackage{leftidx,bbding,extarrows}
\usepackage{tikz,tikz-3dplot}
\usetikzlibrary{matrix,arrows,shapes,decorations.pathmorphing}
 
\usepackage[mathscr]{eucal}
\usepackage{mathrsfs}

\newcounter{tightenum}
\newenvironment{tightitemize}%
{\begin{list}{$\bullet$}{\setlength{\itemsep}{0pt}\setlength{\parsep}{0pt}\setlength{\topsep}{0pt}}}%
{\end{list}}
{\begin{list}{(\roman{tightenum})}{\usecounter{tightenum} \setlength{\itemsep}{0pt}\setlength{\parsep}{0pt}\setlength{\topsep}{0pt}}}%
{\end{list}}

\DeclareMathAlphabet{\mathcal}{OMS}{cmsy}{m}{n}

\makeatletter
\g@addto@macro\bfseries{\boldmath}
\makeatother

\DeclareMathAlphabet\BEuScript{U}{eus}{b}{n}

\newcommand{\CC}{{\mathbb C}}
\newcommand{\NN}{{\mathbb N}}
\newcommand{\RR}{{\mathbb R}}

\usepackage{bm}
\DeclareMathAlphabet{\mathbfsf}{\encodingdefault}{\sfdefault}{bx}{n}

\DeclareBoldMathCommand\Db{D}
\DeclareBoldMathCommand\Fb{F}
\DeclareBoldMathCommand\Ib{I}
\DeclareBoldMathCommand\Lb{L}
\DeclareBoldMathCommand\Mb{M}
\DeclareBoldMathCommand\Nb{N}
\DeclareBoldMathCommand\Pb{P}
\DeclareBoldMathCommand\Ob{0}
\DeclareBoldMathCommand\rb{R}
\DeclareBoldMathCommand\ab{a}
\DeclareBoldMathCommand\bb{b}
\DeclareBoldMathCommand\cb{c}
\DeclareBoldMathCommand\eb{e}
\DeclareBoldMathCommand\ib{i}
\DeclareBoldMathCommand\jb{j}
\DeclareBoldMathCommand\kb{k}
\DeclareBoldMathCommand\pb{p}
\DeclareBoldMathCommand\rb{r}
\DeclareBoldMathCommand\ub{u}
\DeclareBoldMathCommand\vb{v}
\DeclareBoldMathCommand\xb{x}

\newcommand{\sa}{{\sf a}}

\newcommand{\St}{{\mathfrak{S}}}

\newcommand{\id}{\text{id}}

\newcommand{\BB}{\mathscr{B}}

\newcommand{\FF}{\mathscr{F}}
\newcommand{\HH}{\mathscr{H}}

\newcommand{\KK}{\mathscr{K}}

\newcommand{\sN}{\mathscr{N}}

\newcommand{\OO}{{\mathcal{O}}}

\newcommand{\Xc}{{\mathcal{X}}}
\newcommand{\Yc}{{\mathcal{Y}}}

\newcommand{\CoinX}[1]{C_0^\infty({#1})}

\newcommand{\II}{\leavevmode\hbox{\rm{\small1\kern-3.8pt\normalsize1}}}

\newcommand{\ip}[2]{{\langle #1\mid #2\rangle}}

\newcommand{\ket}[1]{{\vert #1\rangle}}
\newcommand{\bra}[1]{{\langle #1 \mid}}

\newcommand{\Tr}{\textrm{Tr}\,}
\newcommand{\supp}{\textrm{supp}\,}

\renewcommand{\Re}{\textrm{Re}\,}

\DeclareMathOperator{\diam}{diam}

\DeclareMathOperator{\Ad}{Ad}

\DeclareMathOperator{\SO}{SO}

\DeclareMathOperator{\pr}{pr}

\newcommand{\Loc}{\mathbfsf{ Loc}}

\newcommand{\CAlg}{\mathbfsf{ C^*\hbox{-}Alg}}

\newcommand{\Af}{{\mathscr A}}

\newcommand{\Sf}{{\mathscr S}}

\newcommand{\ogth}{{\mathfrak o}}
\newcommand{\tgth}{{\mathfrak t}}

\newcommand{\Ngth}{{\mathfrak N}}

\newcommand{\Fgth}{{\mathfrak F}}

\newcommand{\Rgth}{{\mathfrak R}}

\newcommand{\kin}{{\text{kin}}}

\newtheorem{theorem}{Theorem}[section]
\newtheorem{proposition}[theorem]{Proposition}
\newtheorem{definition}[theorem]{Definition}
\newtheorem{lemma}[theorem]{Lemma}

\definecolor{PaleGreen}{rgb}{0.6,1.0.6}
\definecolor{bg}{rgb}{.22,.19,.54}
\definecolor{Navy}{rgb}{.137,.137,.556}
\definecolor{Gold}{rgb}{.93,.82,.24}
\definecolor{vlg}{rgb}{.95,.95,.95}
\definecolor{Green}{rgb}{0,1,0}
\definecolor{LightBlue}{rgb}{0.75,0.75,1}

\begin{document}
\title{The split property for quantum field theories \\ in flat and curved spacetimes}

\author[1]{Christopher J. Fewster\thanks{\tt chris.fewster@york.ac.uk}}
\affil{Department of Mathematics,
University of York, Heslington, York YO10 5DD, United Kingdom.}

\date{\today}
\maketitle

\begin{abstract} 
The split property expresses a strong form of independence of spacelike separated regions in
algebraic quantum field theory. In Minkowski spacetime, it can be proved under
hypotheses of nuclearity. An expository account is given of nuclearity and the split property, 
and connections are drawn to the theory of quantum energy inequalities. In addition, a recent proof of the split property for quantum field theory in curved spacetimes is outlined, emphasising the essential ideas.
\end{abstract}

{\noindent\em Dedicated to the memory of Rudolf Haag}

\section{Introduction}

Special relativity entails that information and influences propagate at speeds no greater than that of light, in order that causes precede effects according to all inertial clocks. Laboratories in spacelike separated spacetime regions should therefore function independently. The \emph{split property} is an expression of
this independence in the algebraic formulation of quantum field theory~\cite{Haag} that is considerably deeper than the assumption of Einstein causality (commutation of spacelike separated observables). 
One aim of this paper is to present a short account of a recent extension of the split property
to locally covariant quantum field theories in curved spacetimes~\cite{Few_split:2015}. 
However, it is also intended to give an expository account of the split property and its consequences,
and also of the hypotheses of \emph{nuclearity} under which the split property was first proved
in a general setting~\cite{BucDAnFre:1987}. This serves both to make apparent the significance of the
split property, and also the physical circumstances in which it holds. Although much of the material in Sections~\ref{sec:split} and~\ref{sec:nuc} reviews other works, the discussion of links between nuclearity and Quantum Energy Inequalities is based on results proved here for the first time. The proof of our main results
in curved spacetime (Section~\ref{sec:cst}) is a streamlined and extended version of arguments going back to Verch~\cite{Verch:1993,Verch_nucspldua:1993}. In fact, it provides a proof strategy for a number
of other results, including the Reeh--Schlieder property~\cite{Sanders_ReehSchlieder} and modular nuclearity~\cite{LecSan:2015}, and is therefore of independent interest.

\section{The split property in Minkowski space}\label{sec:split}
 
To begin, let us consider a configuration of three spacetime regions in Minkowski
spacetime, displayed in Fig.~\ref{fig:regions}. Region $O_1$ is contained within region 
$O_2$, which is \emph{spacelike separated} from region $O_3$. All three regions
are supposed to be open and relatively compact. Spacelike separation means that there
are no causal curves with one endpoint in $O_2$ and the other in $O_3$; in other words, 
there is no possibility of communication between them, if we assume the basic precepts
of special relativity, and it should be possible for experiments in $O_3$ to take place independently of those in region $O_2$. 

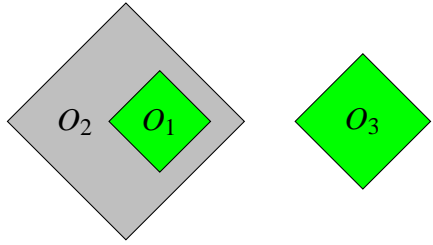
\begin{figure} 
\begin{center}
\begin{tikzpicture}[scale=0.45]
\draw[fill=Green] (7,2) +(-2,0) -- +(0,-2) -- +(2,0) -- +(0,2) -- cycle;
\node at (7,2) {$O_3$};
\draw[fill=lightgray] (0,2) +(-3.5,0) -- +(0,-3.5) -- +(3.5,0) -- +(0,3.5) -- cycle;
\draw[fill=Green] (1,2) +(-1.5,0) -- +(0,-1.5) -- +(1.5,0) -- +(0,1.5) -- cycle;
\node at (-1.5,2) {$O_2$};
\node at (1,2) {$O_1$};
\end{tikzpicture}
\end{center} \caption{Three spacetime regions in Minkowski space. Time runs up the page,
and rays travelling at the speed of light have unit gradient. Thus there is no possibility of communication between regions $O_2$ and $O_3$.}
\label{fig:regions}
\end{figure}

In a theory of local quantum physics~\cite{Haag,AdvAQFT}, each of these regions $O_i$ (and indeed, 
every open relatively compact set) has a corresponding unital $*$-algebra 
$\Rgth(O_i)$ whose self-adjoint elements correspond to observables that can, in principle, be 
measured by experiments conducted within $O_i$. For our present discussion, the $\Rgth(O_i)$
will be von Neumann algebras, concretely represented on a Hilbert space $\HH$, and sharing the
Hilbert space identity operator as their common unit. For the moment, $\HH$ will not
be assumed separable.
As any experiment conducted within $O_1$ is \emph{a fortiori} conducted within $O_2$, 
we require $\Rgth(O_1)\subset \Rgth(O_2)$. This property is known 
as \emph{isotony}. 

How is the `experimental independence' of regions $O_2$ and $O_3$ to be expressed mathematically? 
This question is much more subtle and has many more potential answers than one might first think (see~\cite{Summers:1990,Summers:2009,RedeiSummers:2010}). A starting point is the requirement that measurements can be made independently in $O_2$ and $O_3$. In elementary quantum mechanics one learns that observables are independently measureable
(commensurable) if and only if they commute, so one should certainly require that the algebras
for regions $O_2$ and $O_3$ commute elementwise,
\begin{equation}\label{eq:Einstein}
[\Rgth(O_2),\Rgth(O_3)] = 0,
\end{equation}
which is sometimes called \emph{Einstein causality}.\footnote{Fermi fields, of course, 
anticommute at spacelike separation, and are excluded from the algebra of \emph{observables}, though operators constructed from products of even numbers of Fermi fields qualify as observables.} However, experiments not only concern measurements, but also involve a stage of preparation (corresponding to a choice of a state); what we need, then, is a way of
describing the possibility for experimenters in spacelike separated regions to make both preparations
and measurements independently. 

This is where the split property enters. It turns out to be important to ensure that the regions
in question do not touch at their boundaries, and therefore we switch attention to the independence
of regions $O_1$, surrounded by a `collar region' provided by $O_2$, and $O_3$. 
We describe the inclusion  $\Rgth(O_1)\subset \Rgth(O_2)$ as
\emph{split} if there is a type I von Neumann factor $\Ngth$ such that
\begin{equation}\label{eq:split}
\Rgth(O_1)\subset \Ngth\subset \Rgth(O_2).
\end{equation}
Here, $\Ngth$ is a \emph{factor} if $\Ngth\cap\Ngth'=\CC\II$, where the prime denotes the 
commutant, and the designation as a `type I factor' means that $\Ngth$ is isomorphic as a von Neumann algebra to the algebra
$\BB(\KK)$ of bounded operators on a [not necessarily separable] Hilbert space $\KK$.
The \emph{split property} requires that whenever $O_1\Subset O_2$, i.e., that 
the closure of $O_1$ is compactly contained in $O_2$, then the inclusion
$\Rgth(O_1)\subset \Rgth(O_2)$ is split in this way. Returning to our configuration 
in Fig.~\ref{fig:regions}, the split property has the important consequence that $\Rgth(O_1)$ and $\Rgth(O_3)$ are \emph{$W^*$-statistically independent}:
for all normal states\footnote{A normal state on a von Neumann algebra can be defined
abstractly in terms of its continuity properties; however, when the von Neumann algebra acts on a Hilbert space, the normal states are precisely those that can be represented by density matrix states on the Hilbert space  \cite[Thm 2.4.21]{BratRob}.} $\varphi_1$ 
and $\varphi_3$ on these algebras there is a normal state $\varphi$ 
on the von Neumann algebra $\Rgth(O_1)\vee \Rgth(O_3)$ they generate so that
\begin{equation}
\varphi(AB) = \varphi_1(A)\varphi_3(B) \qquad
A\in\Rgth(O_1), B\in\Rgth(O_3).
\end{equation}
This asserts that any pair of preparations made by the experimenters in regions $O_1$ and $O_3$ can be subsumed into a preparation of the system as a whole. 

The proof of this statement is straightforward and illuminates some related issues. Because $\Ngth$ is a type I factor, there is a Hilbert space isomorphism $U:\HH\to\KK\otimes\KK'$ such that $U\Ngth U^{-1}=\BB(\KK)\otimes\II_{\KK'}$ (see, e.g., \cite[III.1.5.3]{Blackadar}). 
Owing to the inclusion $\Rgth(O_1)\subset\Ngth$ and the fact that $\Rgth(O_3)\subset \Rgth(O_1)'$,
one finds
\begin{align}\label{eq:inclusions}
 U\Rgth(O_1)U^{-1}&\subset \BB(\KK)\otimes\II_{\KK'}, \nonumber\\  
U\Ngth U^{-1}& =\BB(\KK)\otimes\II_{\KK'}, \nonumber\\ 
U\Rgth(O_3)U^{-1} &\subset \II_{\KK}\otimes\BB(\KK').
\end{align}
We now have faithful normal representations $\pi_i$ of $\Rgth(O_i)$ with $\pi_1$ 
acting on $\KK$ and $\pi_3$ on $\KK'$, so that 
\[
\pi_1(A)\otimes\II_{\KK'}=UAU^{-1}, \qquad
\II_\KK\otimes\pi_3(B)=UBU^{-1}
\]
for $A\in\Rgth_1(O)$, $B\in\Rgth(O_3)$. The unitary $U$ clearly implements a spatial
isomorphism of $\Rgth(O_1)\vee \Rgth(O_3)$ with the spatial tensor product $\pi_1(\Rgth(O_1))\bar{\otimes}\pi_3(\Rgth(O_3))$
and, in particular, the map $AB\mapsto A\otimes B$ extends to an
isomorphism of von Neumann algebras 
\begin{equation}\label{eq:tensor}
\Rgth(O_1)\vee \Rgth(O_3) \cong \Rgth(O_1)\bar{\otimes} \Rgth(O_3).
\end{equation}
(In the converse direction, we note that if an isomorphism of this form exists and is spatial, 
then the inclusion $\Rgth(O_1)\subset \Rgth(O_3)'$ is split: simply define $\Ngth$
using the second line of \eqref{eq:inclusions}.) 

Returning to the question of $W^*$-statistical independence and
representing the normal states $\varphi_i$ by density matrices $\rho_i$ so that $\varphi_i(X)=\Tr \rho_i X$, ($X\in\Rgth(O_i)$), we obtain density matrices $\tilde{\rho}_1$ on $\BB(\KK)$ and $\tilde{\rho}_3$ on $\BB(\KK')$ by the partial traces
\begin{equation}
\Tr \tilde{\rho}_1 K=\Tr \rho_1 U^{-1}(K\otimes\II_{\KK'})U, \qquad \Tr \tilde{\rho}_3 K'=\Tr \rho_3 U^{-1}( \II_{\KK}\otimes K')U
\end{equation}
for $K\in\BB(\KK)$, $K'\in\BB(\KK')$.  
Then the desired product state is $\varphi(C) = \Tr \rho C$ with density matrix 
\begin{equation}
\rho= U^{-1}\left(\tilde{\rho}_1\otimes\tilde{\rho}_3\right) U.
\end{equation}
For, taking $A\in\Rgth(O_1)$ and $B\in\Rgth(O_3)$, and writing
$UAU^{-1}=\tilde{A}\otimes\II_{\KK'}$, $UBU^{-1}=\II_{\KK}\otimes \tilde{B}$,
\begin{align}
\Tr \rho AB &= \Tr U^{-1} \left(\tilde{\rho}_1\otimes\tilde{\rho}_3\right) U AB = \Tr  \left(\tilde{\rho}_1\otimes\tilde{\rho}_3\right) \tilde{A}\otimes\tilde{B}=
\left( \Tr \tilde{\rho}_1 \tilde{A}\right)\left(\Tr\tilde{\rho}_3 \tilde{B}\right) \nonumber \\
&=\left( \Tr \rho_1 A\right)\left(\Tr\rho_3 B\right).
\end{align}

The idea that suitable inclusions of local algebras should split is due to Borchers, but the first
proof of the split property (already described as ``an old conjecture of Borchers'') was given in 1974 by Buchholz, in the particular case of the massive free scalar field~\cite{Buc:1974}. There matters rested for some time -- an extension to the observable algebras for the Dirac and Maxwell fields, as well as the massless scalar, was given five years later in \cite{HorDad:1979}, and a related result for the \emph{field algebras} of the massive free fermion fields was given by Summers~\cite{Summers:1982} in 1982.\footnote{Einstein causality must be modified for field algebras of fermionic systems. If $O_2$ and $O_3$ are spacelike separated, the corresponding field algebras $\Fgth(O_i)$ obey a graded commutation relation in place of \eqref{eq:Einstein}; however, by
introducing a suitable unitary twist map $Z$ on $\HH$, and defining the twisted algebra $\Fgth^t(O_3)=
Z\Fgth(O_3)Z^{-1}$, Einstein causality can be reformulated as  $[\Fgth(O_2),\Fgth^t(O_3)]=0$. Then the
relevant form of the split property is that the inclusion $\Fgth(O_1)\subset \Fgth^t(O_3)'$ splits.}
There were then two breakthroughs. First, the theory of split inclusions were studied in a deep
paper of Doplicher and Longo~\cite{DopLon:1984}, showing that the split property has
numerous important consequences, some of which will be discussed below; second, the development of \emph{nuclearity}
criteria~\cite{BucWic:1986} permitted the proof of the split property for general models with
sufficiently good nuclearity properties~\cite{BucDAnFre:1987}. This was a significant step, 
because nuclearity is closely related to questions of thermodynamic stability~\cite{BucWic:1986,BucJun:1986,BucJun:1989} and so it became apparent that -- in 
keeping with Borchers' conjecture -- the split property should indeed be a feature of
suitably well-behaved quantum field theory models. 

The nuclearity criterion used in~\cite{BucDAnFre:1987} is defined as 
follows. First, let us recall that a linear map between Banach
spaces $\Xi:\Xc\to\Yc$ is said to be \emph{nuclear} if there is
a countable decomposition $\Xi(\cdot)=\sum_k  \ell_k(\cdot)\psi_k$ where
$\psi_k\in\Yc$ and $\ell_k\in\Xc^*$, such that the sum $\sum_k \|\psi_k\|\,\|\ell_k\|$ is convergent. Under these circumstances, the \emph{nuclearity index} $\|\Xi\|_1$ is defined as the infimum of the value of this sum over all possible decompositions. 

Now let $O$ be a nonempty open and bounded region of
Minkowski spacetime, with associated von Neumann algebra $\Rgth(O)$. 
Denote the vacuum state vector by $\Omega\in\HH$ and the Hamiltonian, generating time translations, by $H$, so that $U(\tau)=e^{iH\tau}$ satisfies 
\begin{equation}
U(\tau)\Rgth(O) U(\tau)^{-1}= \Rgth(O_\tau),
\end{equation}
where $O_\tau$ is the translation of $O$ under $(t,\xb)\mapsto (t+\tau,\xb)$. The nuclearity criterion of~\cite{BucDAnFre:1987} is that for every $O$ and $\beta>0$, 
each map $\Xi_{O,\beta}:\Rgth(O)\to\HH$ given by
\begin{equation}\label{eq:Xi_def}
\Xi_{O,\beta}(A) = e^{-\beta H} A\Omega
\end{equation}
is nuclear and has nuclearity index obeying
\begin{equation}\label{eq:nuclearity_criterion}
\|\Xi_{O,\beta}\|_1 \le e^{(\beta_0/\beta)^n}
\end{equation}
as $\beta\to 0+$, where the constants $n>0$ and $\beta_0>0$
may depend on $O$ (but not $\beta$). The physical
understanding of these criteria is that $\|\Xi_{O,\beta}\|_1$ plays
the role of a local partition function 
and probes the state space available on given distance and energy scales; see Section~\ref{sec:nuc} for other aspects of the
interpretation.
Note that there is a variety of nuclearity conditions (see~\cite{BucPorr:1990} for discussion) some of which should
be used with care~\cite{FOP}. Given this definition, one has:
 
\begin{theorem}[\cite{BucDAnFre:1987}]\label{thm:nuctosplit}
Suppose that the net of von Neumann algebras $O\mapsto \Rgth(O)$ obeys isotony, that the Hamiltonian is nonnegative, with
zero eigenspace spanned by $\Omega$, and that the above nuclearity 
criterion holds. Then for any open bounded regions with $O_1\Subset O_2$, 
the inclusion $\Rgth(O_1)\subset\Rgth(O_2)$ is split. 
\end{theorem}   

Let us now turn to the theory of standard split inclusions~\cite{DopLon:1984}. Consider
von Neumann algebras $\Rgth_1$ and $\Rgth_2$ acting on a Hilbert space $\HH$. A triple $(\Rgth_1, \Rgth_2,\Omega)$, where $\Omega\in\HH$, is said to be a \emph{standard split inclusion} if $\Rgth_1\subset \Rgth_2$ is split, and $\Omega$ is cyclic and separating for each of $\Rgth_1$, $\Rgth_2$ and $\Rgth_2\wedge \Rgth_1'$ (and hence for their commutants).  In the case where $\Omega$ is the
vacuum vector of a quantum field theory, the latter assumptions are met as a consequence of the
Reeh--Schlieder theorem~\cite{ReehSchlieder:1961}. A number of remarkable results are proved in~\cite{DopLon:1984}. For instance (setting aside a trivial case $\Rgth_2=\CC\II$) one finds that both the von Neumann algebras and the Hilbert space are
substantially constrained: the $\Rgth_i$ are properly infinite and have separable preduals, while the Hilbert space $\HH$ must be separable. 

There is also a canonical choice of type I factor: because $\Rgth_1\subset\Rgth_2$ is split, there is 
(cf.~\eqref{eq:tensor}) a von Neumann algebra isomorphism 
\begin{align}
\phi:\Rgth_1\vee\Rgth_2' &\to \Rgth_1\bar{\otimes}\Rgth_2\nonumber\\
AB' &\mapsto A\otimes B'.
\end{align} 
As $\Omega$ is cyclic and separating for $\Rgth_1\vee\Rgth_2'$ and $\Omega\otimes\Omega$ is
cyclic and separating for $\Rgth_1\bar{\otimes}\Rgth_2$ acting on $\HH\otimes\HH$,
there is, by Tomita--Takesaki theory, a natural choice of unitary $U:\HH\to\HH\otimes\HH$
implementing $\phi$, $UAB'U^{-1}=A\otimes B'$ for all $A\in\Rgth_1$, $B'\in\Rgth_2'$. 
Thus we have
\begin{align}
U\Rgth_1U^{-1} &= \Rgth_1\otimes\II_\HH, \nonumber\\
U\Rgth_2'U^{-1} &= \II\otimes \Rgth_2', \nonumber\\
U\Rgth_2U^{-1} &= \BB(\HH)\otimes\Rgth_2,
\end{align}
and $\Ngth=U^{-1}(\BB(\HH)\otimes\II_\HH) U$ is the canonical choice of an intermediate type I factor.

An important physical application arises as follows. Suppose that there is a unitary group representation $G\owns g\mapsto V(g)$ on $\HH$, acting so that $V(g)\Rgth_i V(g)^{-1}=\Rgth_i$ for each $i=1,2$, $g\in G$. 
Then, defining $W(g)=U^{-1}(V(g)\otimes\II)U$, we have a new unitary representation of $G$
acting on elements $\Rgth_1\vee\Rgth_2'$ by
\begin{equation}
W(g)AB' W(g)^{-1} = U^{-1}(V(g)\otimes\II)(A\otimes B')(V(g)^{-1}\otimes\II)U = V(g)AV(g)^{-1}B'
\end{equation}
for $A\in\Rgth_1$, $B'\in\Rgth_2$. Moreover, $W(g)\in U^{-1}(\BB(\HH)\otimes\II)U\subset\Rgth_2$. 
In the context of a quantum field theory obeying the split property, this corresponds immediately
to the situation in which $G$ is a group of global gauge transformations, implemented by $V(g)$.
Given any nested pair $O_1\Subset O_2$ one may then obtain a localised representation
$W(g)\in\Rgth(O_2)$ which agrees with $V(g)$ on $\Rgth(O_1)$ but acts trivially on $\Rgth(O_2)'$. 
This again emphasises the way in which the split property permits the physics of
the inner region to be isolated from that of the causal complement of the outer region.
The generators of $W(g)$ can be interpreted as suitable smearings of a conserved local current associated to the global symmetry, thus providing an abstract version of Noether's theorem (note that
there is no assumption that the theory derives from a Lagrangian)~\cite{DopLon:1983}.

There is an interesting result in the converse direction. Suppose $\Rgth_1\subset\Rgth_2$, 
and that the flip automorphism $\sigma:A\otimes B\mapsto B\otimes A$ of $\Rgth_1\otimes\Rgth_1$
is inner with respect to $\Rgth_2\otimes\Rgth_2$ -- i.e., $\sigma$ agrees on $\Rgth_1\otimes\Rgth_1$ with $\Ad U$ for some
unitary $U\in\Rgth_2\otimes\Rgth_2$. Then one has
\begin{equation}
U(\II\otimes AB')U^{-1} = U(\II\otimes A )U^{-1} U(\II\otimes B')U^{-1} = A\otimes B'
\end{equation}
for $A\in\Rgth_1$, $B'\in\Rgth_2'$, where we use the fact that $\Ad U$ implements the flip on 
$\Rgth_1\otimes\Rgth_1$ and (as $U\in\Rgth_2\otimes\Rgth_2$)  acts trivially on $\II\otimes B'
\in\Rgth_2'\otimes\Rgth_2'$. This establishes a von Neumann algebra isomorphism $\Rgth_1\vee\Rgth_2'\to\Rgth_1\bar{\otimes}\Rgth_2'$
extending $AB'\mapsto A\otimes B'$.  If, in addition, there is a cyclic and separating vector
for $\Rgth_1'\wedge\Rgth_2$, then one can establish a spatial
isomorphism between $\Rgth_1\vee\Rgth_2'$ and $\pi_1(\Rgth_1)\bar{\otimes}\pi_2(\Rgth_2)$, where the $\pi_i$ are
faithful normal representations of the $\Rgth_i$ with cyclic and separating vectors.
It follows that the inclusion $\Rgth_1\subset\Rgth_2$ is split (moreover, under these 
circumstances, if the inclusion is split, then the flip is inner with respect to $\Rgth_2\otimes\Rgth_2$)~\cite{DAntLon:1983,DopLon:1984}. Clearly, the essential
point here is that $\sigma(1\otimes A)=A\otimes 1$, and the argument works for
more general automorphisms than just the flip.

D'Antoni and Longo used this idea in an ingenious proof of the split
property for the free scalar field. Let $O\mapsto\Rgth(O)$ be the
net of von Neumann algebras of the free massive scalar field on Minkowski
space in the vacuum representation. Then $O\mapsto\Rgth(O)\bar{\otimes}\Rgth(O)$ 
is the theory of two independent massive scalar fields, and has an internal $\SO(2)$ symmetry
group that rotates the doublet of fields, with a rotation through $\pi/2$ precisely inducing $\II\otimes A\mapsto A\otimes \II$. The gauge symmetry is associated with a Noether current, and 
it is then shown that suitable local smearings of these currents in a region $O_2$ 
generate a local implementation of the gauge symmetry on a smaller region $O_1\Subset O_2$
and acting trivially on $\Rgth(O_2)'$. The vacuum provides the cyclic and separating vector to permit the deduction that $\Rgth(O_1)\subset\Rgth(O_2)$ is split.
A similar idea has been employed recently~\cite{MorTom:2010} to make a more explicit
local implementation of the generators of gauge symmetries -- actually, a family of
possible implementations are constructed, with the `canonical' implementation
described above included as a special case. 
With appropriate modifications, local implementations of geometric symmetries can also 
be constructed~\cite{BucDopLon:1986,Carpi:1999}. 
 
One of the most striking results to emerge from the body of work on the split property
concerns the type of the local von Neumann algebras~\cite{BucDAnFre:1987}. In 1985, Fredenhagen showed that, under the hypothesis of a scaling limit, local algebras are of type III${}_1$
(for brevity, we assume here for that the local algebras are factors). 
Combining this with the split property and the Reeh--Schlieder theorem, and the
assumption that any $\Rgth(O)$ may be generated as $\Rgth(O)=\bigvee_k \Rgth(O_k)$ 
for some increasing chain of nested sets $O_k\Subset O_{k+1}$, one then has
$\Rgth(O)=\bigvee_k \Ngth_k$ for a sequence of nested type I factors. As the
Hilbert space is separable, this leads to the conclusion that each $\Rgth(O)$ is
a hyperfinite type III${}_1$ factor, thus fixing it uniquely up to isomorphism~\cite{Haagerup:1987}.
This is a remarkable achievement, demonstrating that the distinction between
different quantum field theories lies in the relationships between local algebras, rather than
the content of those algebras \emph{per se}. The timing of the results is also remarkable, with progress on the QFT side at exactly the moment that the structural results on von Neumann algebras appeared. 
  
Finally, let us note that the split property plays a key role in the construction of
certain integrable quantum field models -- see~\cite{Lech_chap:2015} for a review,
and that the interpretative framework of models described by funnels of type I factors
has been investigated in~\cite{BucSto:2014}.
 
\section{Nuclearity and Quantum Energy Inequalities}\label{sec:nuc}

The nuclearity criterion, while being physically well-motivated, 
is rather technical. In this section we explain more about its 
physical status and draw some links to the theory of quantum energy inequalities. For the purposes of illustration, we consider a 
class of theories, consisting of countably many independent 
free scalar fields with masses $m_r$ ($r\in\NN$) in $4$ spacetime dimensions --- this may also be related to a particular generalised free field~\cite{DopLon:1984}. For simplicity, we assume that $m_r$ form a nondecreasing sequence and that there is a mass gap, i.e., $m_1>0$. 

The models are constructed as follows. Each individual field has
a Hilbert space $\FF_r$ which is a copy of the symmetric Fock space
over $\HH=L^2(\RR^3,d^3\kb/(2\pi)^3)$, with vacuum vector $\Omega_r$. The annihilation and creation operators obey
\begin{equation}
[\sa_r(u),\sa_r(v)]=0, \qquad [\sa_r(u),\sa_r^*(v)]=\ip{u}{v}\II_{\FF_r}, 
\qquad (u,v\in \HH)
\end{equation}
and $\sa_r(u)\Omega_r=0$ for all $u\in\HH$.  
The smeared quantum field is
\begin{equation}
\Phi_r(f):= \sa_r(K_r\bar{f}) + \sa_r(K_r f)^*,
\end{equation}
where $K_r:\CoinX{\RR^4}\to \HH$ is defined by $(K_r f)(\kb)=(2\omega_r)^{-1/2}\hat{f}(\omega_r(\kb),\kb)$, with the Fourier transform $\hat{f}(k)=\int d^4x\, e^{ik\cdot x} f(x)$ 
and $\omega_r(\kb)=(\|\kb\|^2+m_r^2)^{1/2}$. Here
$k$ is a $4$-vector, and the $\cdot$ denotes the Lorentz contraction
in the $+---$ signature. For real-valued $f$, the operators
$\Phi_r(f)$ are essentially self-adjoint on a domain of finite particle
vectors and we use $\Phi_r(f)$ also to denote the closure. The
local von Neumann algebras $\Rgth_r(O)$ are then generated by the Weyl operators $\exp(i\Phi_r(f))$ as $f$ runs over real-valued elements of $\CoinX{O}$. The combined theory lives on the incomplete tensor product $\FF=\otimes_r\FF_r$ relative to $\Omega=\otimes_r \Omega_r$ and the overall local algebras $\Rgth(O)$ are formed by taking
the weak closure in $\FF$ of the algebraic tensor product $\odot_r\Rgth_r(O)$.    

For these models, we may give necessary and sufficient criteria
for nuclearity. 
\begin{proposition} Consider the theory of countably many scalar fields with mass gap. \\(a) A necessary condition for the theory to obey the nuclearity criterion  is that
\begin{equation}\label{eq:nuc_nec}
F(\beta):=\sum_r \frac{e^{-4\beta m_r}}{m_r^2}  
\end{equation} 
should be finite for all $\beta>0$, with $\log F(\beta)$ growing at most polynomially in $\beta^{-1}$ as $\beta\to 0+$.\\
(b) A sufficient condition for the theory to obey the nuclearity criterion  is that  
\begin{equation}
G(\beta):=\sum_r e^{-\beta m_r/4}
\end{equation}
should be finite for all $\beta>0$ and grows
only polynomially in $\beta^{-1}$ as $\beta\to 0+$. 
\end{proposition}
\begin{proof}[Outline of the proof]  To begin, let us develop a lower bound on the nuclearity index following~\cite[\S 2]{BucPorr:1990}. Suppose $\Xi_{O,\beta}$ is
nuclear with a decomposition of the form $\sum_k  \ell_k(\cdot)\psi_k$,
and note that the operator norm obeys $\|\Xi_{O,\beta}\|\le 1$ by \eqref{eq:Xi_def}.  
Suppose further that one can find a sequence $A_r\in\Rgth(O)$
so that the $\Xi_{O,\beta}(A_r)$ are orthogonal in $\FF$, 
assuming without loss that $\|A_r\|\le 1$. Then
\begin{equation}
\sum_r \|\Xi_{O,\beta}(A_r)\|^4 =\sum_r \left|\sum_{k} \ell_k(A_r)\ip{\Xi_{O,\beta}(A_r)}{\psi_k}\right|^2 = \left\|\sum_k B_k \psi_k\right\|^2 
\le\left(\sum_k \|\ell_k\|\, \|\psi_k\|\right)^2
\end{equation}
using orthogonality and noting that $B_k:=\sum_r \ell_k(A_r) \|\Xi_{O,\beta}(A_r)\|^{-1} \ket{\Xi_{O,\beta}(A_r)}\bra{\Xi_{O,\beta}(A_r)}$ has norm $\|B_k\|\le \|\ell_k\|$, given our assumptions on the norms of $A_r$ and $\Xi_{O,\beta}$. 
Considering all possible decompositions, this shows that
\begin{equation}
\sum_r \|\Xi_{O,\beta}(A_r)\|^4 \le   \|\Xi_{O,\beta}\|_1^2.
\end{equation} 

Fixing $O$, Buchholz and Porrmann~\cite{BucPorr:1990} construct a sequence $A_r$ in $\Rgth(O)$, built from Weyl operators, so that the $\Xi_{O,\beta}(A_r)$ are orthogonal in $\FF$ and  $\| \Xi_{O,\beta}(A_r)\|^4 \ge C e^{-4\beta m_r}/m_r^2$
for sufficiently large $r$. This establishes part (a). 

On the other hand, we note that in the full  theory
the nuclear maps $\Xi_{O,\beta}$ are simply tensor products of
the maps for each $m_r$, and that the nuclearity index of the 
combined theory is therefore bounded above by the product of the
nuclearity indices of the individual theory. An involved computation (see \cite[\S 5]{BucWic:1986}
for the original version in a slightly different formulation of nuclearity, or \cite[\S 17.3]{BaumWollen:1992} 
in the present version) furnishes an upper bound
\begin{equation}
\|\Xi_{O,\beta}\|_1 \le \exp \left\{c\frac{R^3}{\beta^3}\sum_r |\log(1-e^{-\beta m_r/2})|\right\}
\end{equation}
if $O$ is a diamond whose base is a ball of radius $R>m_1^{-1}$, and where
$c>0$ is a numerical constant independent of $r,\beta,m_r$. 
Using the fact that $\sup_{x>0} x e^{x/2}|\log(1-e^{-x})|<\infty$, this bound can be estimated above
to give
\begin{equation}
\|\Xi_{O,\beta}\|_1 \le \exp \left\{C\frac{R^3}{m_1\beta^4}\sum_r e^{-\beta m_r/4}\right\}
\end{equation}
and part (b) is thereby proved.  
\end{proof}
In particular, one sees that the mass spectrum
\begin{equation}\label{eq:bad_mn}
m_r=(2d_0)^{-1}\log(r+1)
\end{equation}
for fixed $d_0>0$ is inconsistent with nuclearity, because $F(\beta)$ is infinite for $\beta<\frac{1}{2}d_0$, while
a mass spectrum $m_r=r m_1$ is easily seen to satisfy the sufficient
condition for nuclearity. 

More insight is obtained by defining the counting function
\begin{equation}\label{eq:Nu}
N(u) = \sum_{r} \vartheta(u-m_r)
\end{equation}
of fields with mass below $u$, for we have
\begin{equation}
G(\beta)=\int_0^\infty e^{-\beta u/4}dN(u) = \frac{\beta}{4}\int_0^\infty e^{-\beta u/4}N(u)\,du.
\end{equation}
Thus by a direct computation, polynomial boundedness of $N(u)$
as $u\to\infty$ is sufficient for nuclearity. Note that the example given by \eqref{eq:bad_mn} corresponds to an exponentially growing counting function. On the other hand, if nuclearity holds then we have 
$\sum_r e^{-\beta m_r/2}/m_r^2\le e^{(\beta_0/\beta)^n}$
as $\beta\to 0+$ (we have absorbed some constants into $\beta_0$, 
but the value of $n$ is as in the nuclearity criterion \eqref{eq:nuclearity_criterion}). Thus  
\begin{equation}
\sum_r e^{-\beta m_r } =
\sum_r m_r^2 e^{-\beta m_r/2 }\frac{e^{-\beta m_r/2 }}{m_r^2} \le \frac{A}{\beta^2} e^{(\beta_0/\beta)^n}
\end{equation}
as $\beta\to 0+$ for some $A>0$ and therefore  
\begin{equation}
N(v) \le \sum_r e^{\beta (v-m_r) }
\le  \frac{A}{\beta^2} e^{(\beta_0/\beta)^{n}+ \beta v}
\end{equation}
for any $v,\beta>0$. We are free to optimise the right-hand
side by choice of $\beta$.\footnote{This is a well-known
technique in the theory of Tauberian estimates of sums and integrals, see e.g., ~\cite{Odlyzko:1992}.} An exact minimisation of the right-hand side over $\beta$ is awkward,  but we may certainly substitute $\beta=\beta_0^{n/(n+1)} (n/v)^{1/(n+1)}$ [which minimises the exponential factor] yielding
\begin{equation}\label{eq:Nbound}
N(v) \le B v^{2/(n+1)} e^{C v^{n/(n+1)}}
\end{equation}
for positive constants $B$ and $C$. Therefore, nuclearity implies
a sub-exponential growth in $N$, and is implied by polynomial growth. 

It is important to understand that the criteria just given are correlated with physical properties.
For example, Buchholz and Junglas~\cite{BucJun:1986} showed that convergence
of the sum $\sum_{n} e^{-\beta m_n/2}$ is sufficient for 
the thermal equilibrium state $\omega_\beta$ at inverse temperature $\beta$ to be locally normal\footnote{That is, its restriction to any local von Neumann algebra (in the vacuum representation) of a relatively compact region is normal.} to the vacuum state; on the other hand, if 
$\omega_\beta$ is locally normal, then $\sum_{n} e^{-2\beta m_n}$ converges. 
Therefore, local normality for all $\beta>0$ is equivalent to finiteness of $F(\beta)$,
which indicates a link between nuclearity criteria and the good thermodynamic behaviour. 
This suggestion was made precise in the general context, also by Buchholz and Junglas~\cite{BucJun:1989}.

Theories in which $F(\beta)$ diverges for sufficiently small $\beta$
have interesting behaviour relative to the split property. For
example, the mass spectrum~\eqref{eq:bad_mn} produces
a theory in which one has split inclusions of $\Rgth(O_1)\subset\Rgth(O_2)$ only when $O_2$ is sufficiently
larger than $O_1$: the so-called \emph{distal split property}. 
To be precise,  
consider the case where $O_i$ ($i=1,2$) are Cauchy developments
of concentric open balls with radii $r_i$ lying in a common spacelike hyperplane. Then one defines a \emph{splitting distance} $d(r_1)$ to be the infimum of $r>0$ for which \eqref{eq:split} holds with $r_2 = r+r_1$.
In the model~\eqref{eq:bad_mn} it can be shown~\cite[Thm 4.3]{DAnDopFreLon:1987} that the splitting distance obeys $d_0\le d(r) \le 2d_0$ for all $r>0$. Thus the 
inverse splitting distance is $d(r)^{-1}$ of the same order as the maximum temperature for which locally normal equilibrium states exist.  

A second (and more quantitative) illustration of the significance of nuclearity criteria is provided by quantum energy inequalities (QEIs). In classical field theory, the models typically studied have everywhere nonnegative energy densities according to inertial observers. By contrast, no quantum field theory obeying the
standard assumptions can admit an energy density with nonnegative 
expectation values in all physical states (and vanishing in the vacuum state)~\cite{EGJ}. However, as first suggested by Ford~\cite{Ford78}, 
the extent to which energy densities can remain negative turns out
to be constrained by bounds -- the quantum energy inequalities --
in a number of models (see~\cite{Few:Bros,Fews_AEIlectures:2012} for
reviews). For instance, the (Wick ordered) energy density $\rho_m$ for a single free scalar field of mass $m$ obeys a lower bound 
\begin{equation}
\int \ip{\Psi}{\rho_m(t,\xb)\Psi} |g(t)|^2\,dt \ge -C \int_m^\infty u^d |\hat{g}(u)|^2\,du
\end{equation}
for all normalised Hadamard vector states $\Psi$, where $g\in\CoinX{\RR}$ 
and the constant $C$ depends on the spacetime dimension $d$, but not on $m$, $g$ or $\Psi$ \cite{FewsterEveson:1998,Fews00}.\footnote{Much more general results are obtained in~\cite{Fews00}. The Fourier transform is defined here by $\hat{g}(u)=\int dt\,e^{-iut}g(t)$.} 
The convergence of the integral follows because $\hat{g}(u)$ decays
faster than any inverse power. 
Therefore, the combined energy density $\rho$ of the theory with countably many fields obeys
\begin{equation}\label{eq:towerQEI}
\int \ip{\Psi}{\rho(t,\xb)\Psi} |g(t)|^2\,dt \ge -C \int_0^\infty u^d N(u) |\hat{g}(u)|^2\,du,
\end{equation}
where $N$ is defined in \eqref{eq:Nu} and $\Psi$
is any state in the space $\St$ of finite
linear combinations of tensor product states $\bigotimes_j\psi_j$ 
in which all but finitely many of the $\psi_j$ are in the vacuum state
(for mass $m_j$) and so that the $\psi_j$ are all Hadamard. Thus a theory in which $N$ has polynomial growth (hence obeying nuclearity) produces a well-behaved QEI, with a finite lower bound for any $g\in\CoinX{\RR}$. On the other hand, 
if $N$ grows exponentially (so the theory fails to obey nuclearity) then the bound is divergent for every nontrivial $g\in\CoinX{\RR}$.\footnote{If the
transform $\hat{g}$ decays exponentially then $g$ extends to an
analytic function in a neighbourhood of the real axis; as $g$ is compactly
supported, it would then follow that $g\equiv 0$.}  Moreover, 
specialising to $d=4$, the following result is proved in the Appendix. 
\begin{theorem}\label{thm:mbound}
Let $m_0\ge 0$ be fixed.  Suppose that $f\in\CoinX{\RR}$ is nonnegative, even, and has a Fourier transform (which is necessarily real and even) that is  also nonnegative and bounded from below on $[m_0,\infty)$  
\begin{equation}\label{eq:fhat}
\hat{f}(u)\ge \varphi(|u|),
\end{equation}
where $\varphi:[m_0,\infty)\to\RR^+$ is monotone decreasing. 
Then the Klein--Gordon field of mass $m>m_0$ admits a Hadamard
state, given by a normalised Fock-space vector $\Psi_m$, such that 
\begin{equation}
\int \ip{\Psi_m}{\rho_m(t,\Ob)\Psi_m} f(t)\,dt \le -\Gamma m^4
\varphi(2\sqrt{2}m)^2,
\end{equation}
with a constant $\Gamma>0$ that depends neither on $m$ nor $\varphi$.
\end{theorem}
  
Returning to our model of countably many fields, we deduce
that 
\begin{equation}
\inf_{\substack{\Psi\in\St\\ \|\Psi\|=1}} \int \ip{\Psi}{\rho(t,\Ob)\Psi} f(t)\,dt 
\le -\Gamma \sum_{r} m_r^4 \varphi(2\sqrt{2} m_r)^2,
\end{equation}  
assuming for simplicity that $m_1>m_0$. 
Noting that the rescaled function $f_\lambda(t)=\lambda^{-1}f(t/\lambda)$ obeys
\eqref{eq:fhat} with $\varphi(|u|)$ replaced by $\varphi(\lambda|u|)$ and $m_0$ by $m_0/\lambda$,
the existence of a QEI for this model then implies the convergence of $\sum_r m_r^4 \varphi(2\sqrt{2}\lambda m_r)^2$
for all $\lambda>0$.
 
This observation produces a link between QEIs, thermal stability
and nuclearity. To do this, we first construct some test functions
obeying the hypotheses of Theorem~\ref{thm:mbound}. Observe that if $\chi\in\CoinX{\RR}$ is even and nonnegative, then the convolution $\eta = \chi\star\chi$ is even, nonnegative and has nonnegative Fourier transform. Choose $\beta_0>0$ and define 
\begin{equation}
f(t) = \frac{\beta_0\eta(t)}{\pi(t^2+\beta_0^2)}\,,
\end{equation}
assuming without loss that $\int f(t)\,dt=1$. Then,
for any $u>0$, 
\begin{equation}
\hat{f}(u) = \int_{-\infty}^{\infty}\frac{du'}{2\pi} \hat{\eta}(u')e^{-\beta_0|u-u'|}
\ge e^{-\beta_0 u}\int_{-\infty}^0 \frac{du'}{2\pi} \hat{\eta}(u')
e^{\beta_0 u'}
\end{equation}
so--as $\hat{f}$ is even--we conclude that 
$\hat{f}(u)\ge \kappa e^{-\beta_0|u|}$ for some positive constant $\kappa$.

\begin{theorem}
Let $f\in\CoinX{\RR}$ be a fixed test function obeying the 
hypotheses of Theorem~\ref{thm:mbound} with 
$\varphi(u)=\kappa e^{-\beta_0 u}$ (for some $\kappa, \beta_0>0$)
and define $f_\lambda(t)=\lambda^{-1}f(t/\lambda)$. Consider
the vacuum representation of the theory of countably many independent scalar fields with 
a mass gap and suppose it satisfies a QEI bound 
\begin{equation}
\inf_{\substack{\Psi\in\St\\ \|\Psi\|=1}} \int \ip{\Psi}{\rho(t,\xb)\Psi} f_\lambda(t) \,dt \ge -Q(\lambda) >-\infty
\end{equation}
for all $\lambda>0$. Then the thermal equilibrium states of the theory are locally normal at all temperatures. Furthermore, if $Q$ obeys a polynomial
scaling bound $Q(\lambda)\le C\lambda^{-n}$ for some $C, n>0$
then the theory fulfils the nuclearity criterion and hence has the split property. 
\end{theorem} 
\begin{proof}
The existence of a QEI for this model implies the convergence of $\sum_r m_r^4 e^{-4\sqrt{2}\beta m_r}$, and hence $\sum_r   e^{-4\sqrt{2}\beta m_r}$,  for all $\beta>0$; moreover, 
the scaling bound on $Q(\lambda)$ as $\lambda\to 0+$ immediately implies an inverse polynomial scaling bound on the function $G(\beta)$ 
as $\beta\to 0^+$. Thus,  the nuclearity criterion holds and the split
property follows by Theorem~\ref{thm:nuctosplit}. 
\end{proof}

We remark that this result uses only rather minimal information
on the mass spectrum derived from the existence of QEIs for test functions decaying more
slowly than exponentially. One may also construct compactly supported test functions with more finely controlled decay: moderately explicit examples of smooth, even, everywhere positive functions of compact support $f$ with transforms obeying
\begin{equation}
\hat{f}(u) =   \kappa e^{-\gamma|u|^\alpha} +
O\left(\frac{e^{-\gamma|u|^\alpha}}{|u|^{1-\alpha}}\right) ,
\end{equation}
on $u\neq 0$, where $\kappa, \gamma>0$ and $0<\alpha<1$ are given in \cite[\S II]{FewsFord:2015}; an older but less explicit construction of such functions appears in~\cite{Ingham:1934}. For the theory to obey a finite QEI for all such test functions, it is necessary that $\sum_r m_r^4 e^{-(\beta m_r)^\alpha}<\infty$ for all $\beta>0$ and $0<\alpha<1$.

Connections between QEIs and thermal stability have been studied before. Indeed, arguments based on the second law of thermodynamics underlay 
Ford's original intuition that QFT might obey QEI-type
bounds~\cite{Ford78}. In an abstract setting~\cite{FewVer-Passivity},  it was shown that the existence of QEIs entailed the existence of passive states, thereby proving that
the second law of thermodynamics holds. 
The fact that QEIs with polynomial scaling implies nuclearity was stated without proof in~\cite{Few:Bros}. By way of a converse, we give the following:
\begin{theorem}
If the theory of countably many scalar fields with mass gap satisfies
the nuclearity criterion \eqref{eq:nuclearity_criterion} then one has a QEI bound of the form \eqref{eq:towerQEI} for all $g\in\CoinX{\RR}$ such that  $\hat{g}(u) = O(e^{-\gamma |u|^\alpha})$ for any $\alpha>n/(n+1)$ and $\gamma>0$ (or with $\alpha=n/(n+1)$ and sufficiently large $\gamma$).
\end{theorem}
\begin{proof}
We have already shown that $N(v) \le B v^{2/(n+1)} e^{C v^{n/(n+1)}}$
for positive constants $B$ and $C$. 
Substituting into \eqref{eq:towerQEI}, the
QEI bound is finite for all $g\in\CoinX{\RR}$ meeting the  hypotheses. 
\end{proof}

Note that the conditions on $\hat{g}$ become more stringent as 
$n$ increases (as one would expect), but nonetheless, the result is
nonempty for any $n>0$. As the upper bound \eqref{eq:Nbound} on $N(u)$ arising from nuclearity is rather weak, it seems reasonable
to conjecture that nuclearity for such models implies the existence of a QEI for a wide class of test functions, with scaling behaviour not 
much worse than polynomial.

Do all theories obeying nuclearity satisfy QEIs (or, even more, QEIs with good scaling behaviour)? One problem with this is that there are arguments suggesting that QEIs do not hold for general interacting theories in 
the form we have stated~\cite{OlumGraham03}. Allowing for a lower bound that depends on the overall energy scale of the state (but less strongly than any possible \emph{upper} bound) 
one can establish results analogous to QEIs for `classically positive' observables
in a model independent setting~\cite{BostelmannFewster09}. One drawback to those
results, however, is that it has not yet been possible to show that
the energy density is classically positive in general theories.  
 
Summarising this section: we have explained how nuclearity criteria
are related both to good thermodynamic behaviour
and also to the QEIs. The connection with QEIs is currently restricted
to models consisting of countably many independent free
fields; however the circumstantial evidence for a more general
result seems strong.

\section{The split property in curved spacetime}
\label{sec:cst}

The previous sections have described the split property, 
and the related nuclearity criterion, in flat spacetime. 
Here, we turn to the question of what can be said in
curved spacetimes. While the
split property can be stated in much the same way in curved spacetimes as in flat,
the same is not true of nuclearity because general curved spacetimes do not admit time-translation
symmetries and therefore possess no Hamiltonian.  

In the first instance, then, we restrict to ultrastatic spacetimes. 
Let $(\Sigma,h)$ be a complete Riemannian metric space, 
assumed connected. Then the manifold $\RR\times\Sigma$
equipped with Lorentz-signature metric
\begin{equation}\label{eq:ultra}
g = dt\otimes dt - \pr_\Sigma^* h
\end{equation}
and time-orientation chosen so that $t$ increases to the future, 
is by definition an ultrastatic spacetime. Moreover, the spacetime is 
globally hyperbolic~\cite[Prop. 5.2]{Kay1978}: it admits no
closed timelike curves, and every set of the form $J^+(p)\cap J^-(q)$ is compact, where $J^{+/-}(p)$ is the set of points that can be
reached by future/past-directed smooth causal curves starting 
from $p$. 

Let us assume that a net of local von Neumann algebras, acting on a Hilbert space $\HH$, has been associated to this spacetime, and that the geometrical time-translation symmetry $(t,\sigma)\mapsto (t+\tau,\sigma)$ is implemented by a unitary group $U(\tau)=e^{iH\tau}$, where the generator $H$ is a positive unbounded operator with a one-dimensional kernel spanned by a vector $\Omega$. In this setting, one can formulate nuclearity criteria just as in Minkowski space, and 
derive the split property as a consequence by \cite[Prop.~17.1.4]{BaumWollen:1992} (which is an abstract version of~\cite{BucDAnFre:1987}). This strategy was adopted by 
Verch~\cite{Verch_nucspldua:1993} in the case of the free
Klein--Gordon field on ultrastatic spacetimes, and somewhat later by D'Antoni and Hollands
for Dirac fields~\cite{DAnHol:2006}.

As already mentioned, one cannot pursue this strategy in general spacetimes. Instead one proceeds
indirectly by deforming spacetimes in which the split property is known to hold into others, in such a way that the split property is preserved. This approach
was used by Verch in~\cite{Verch_nucspldua:1993} for the specific
case of the Klein--Gordon field, but has even older
antecedents in other contexts~\cite{FullingNarcowichWald}.

We will present a modern, general and streamlined version of this argument, 
based on \cite{Few_split:2015}, to which the reader is referred for the full
details. The general framework used is that
of locally covariant quantum field theory~\cite{BrFrVe03} (see 
\cite{FewVerch_aqftincst:2015} for a recent expository review).  
The geometric aspects are as follows. Fixing a spacetime 
dimension $d\ge 2$, we define a category of spacetimes $\Loc$, 
whose objects are globally hyperbolic spacetimes of dimension $d$, equipped with 
orientation and time-orientation and typically denoted $\Mb,\Nb$ etc. 
The morphisms in this category
are smooth isometric embeddings, preserving orientation and time-orientation
and having causally convex images -- that is, every causal curve joining
points in the image lies entirely within it. 
If the image of $\psi:\Mb\to\Nb$  contains a Cauchy surface
of $\Nb$, $\psi$ is called a \emph{Cauchy  morphism}. 

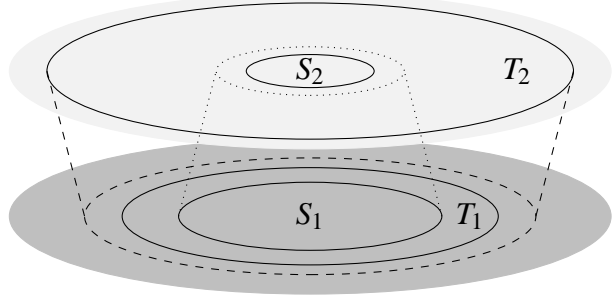
\begin{figure}
\tdplotsetmaincoords{75}{90}
\pgfmathsetmacro{\rvec}{.8}
\pgfmathsetmacro{\thetavec}{15}
\pgfmathsetmacro{\phivec}{60}
\begin{center}
\begin{tikzpicture}[scale=5,tdplot_main_coords]
\coordinate (O) at (0,0,0);
\coordinate (Q) at (0,0,0.4);
\tdplotdrawarc[color=lightgray,fill=lightgray]{(O)}{0.8}{0}{360}{anchor=north}{}
\tdplotdrawarc[color=vlg,fill=vlg]{(Q)}{0.8}{0}{360}{anchor=north}{}
\tdplotdrawarc{(O)}{0.35}{0}{360}{anchor=north}{}
\tdplotdrawarc{(O)}{0.5}{0}{360}{anchor=north}{}
\tdplotdrawarc[dashed]{(O)}{0.6}{0}{360}{anchor=north}{}
\tdplotdrawarc{(Q)}{0.17}{0}{360}{anchor=north}{}
\tdplotdrawarc[dotted]{(Q)}{0.25}{0}{360}{anchor=north}{}
\tdplotdrawarc{(Q)}{0.7}{0}{360}{anchor=north}{}
\draw[dotted] (0,0.35,0) -- (0,0.25,0.4);
\draw[dotted] (0,-0.35,0) -- (0,-0.25,0.4);
\draw[dashed] (0,0.6,0) -- (0,0.7,0.4);
\draw[dashed] (0,-0.6,0) -- (0,-0.7,0.4);
\node at (0,0.0,0.0) {$S_1$};
\node at (0,0.425,0) {$T_1$};
\node at (0,0.55,0.4) {$T_2$};
\node at (0,0.0,0.4) {$S_2$};
\end{tikzpicture}\caption{Regular Cauchy pairs with $(S_1,T_1)\prec_\Mb (S_2,T_2)$.}
\label{fig:preorder}
\end{center}
\end{figure} 
 
We also introduce the following terminology:
\begin{definition}  
A \emph{regular Cauchy pair} $(S,T)$ in $\Mb\in\Loc$
is an ordered pair of nonempty, open, relatively compact subsets of a common smooth spacelike Cauchy surface in which $\overline{T}$ has nonempty complement, and so that $\overline{S}\subset T$.
There is a preorder on regular Cauchy pairs in $\Mb$ defined so that $(S_1,T_1)\prec_\Mb (S_2,T_2)$
if and only if $S_2\subset D_\Mb(S_1)$ and $T_1\subset D_\Mb(T_2)$. 
\end{definition}
Here, $D_\Mb(S)$ is the \emph{Cauchy development} of $S$ in $\Mb$ -- all points $p$ in $\Mb$ such that every inextendible piecewise-smooth causal curve through $p$ intersects $S$. The preordering is illustrated in Fig.~\ref{fig:preorder}.  

Regular Cauchy pairs have important stability properties~\cite[Lem.~2.4]{Few_split:2015}: if $\psi:\Mb\to\Nb$ is Cauchy, then $(S,T)$ is a 
regular Cauchy pair in $\Mb$ if and only if $(\psi(S),\psi(T))$ is a regular Cauchy pair in $\Nb$  and $\overline{\psi(T)}\subset\psi(\Mb)$. Furthermore, in all `sufficiently nearby' Cauchy surfaces to one containing $(S,T)$,
there are regular Cauchy pairs preceding and preceded by $(S,T)$ in the preorder. (See~\cite[Lem.~2.6]{Few_split:2015} for the precise details and more general statements.) 

Locally covariant QFT provides an axiomatic framework for quantum field theory on curved spacetimes
and has led to a number of interesting developments, as described in~\cite{FewVerch_aqftincst:2015}.
The first assumption is: 
\begin{description}
\item[\em Functoriality] A theory is a covariant functor $\Af:\Loc\to \CAlg$ from 
$\Loc$ to the category $\CAlg$ of unital $C^*$-algebras and injective unit-preserving $*$-homomorphisms.
\end{description}
Thus, to each spacetime $\Mb$ there is an object $\Af(\Mb)$ of $\CAlg$ and 
to every morphism between spacetimes $\psi:\Mb\to\Nb$ there is a $\CAlg$-morphism
$\Af(\psi):\Af(\Mb)\to \Af(\Nb)$ such that $\Af(\id_\Mb) = \id_{\Af(\Mb)}$ and
$\Af(\phi\circ\psi) = \Af(\phi)\circ \Af(\psi)$. Given this, we may define
a  \emph{kinematic net} in each $\Mb\in\Loc$ indexed by nonempty open causally 
convex subsets, 
\begin{equation}
O \mapsto \Af^\kin(\Mb;O) := \Af(\iota_{O})(\Af(\Mb|_O)).
\end{equation}
Here $\Mb|_O$ is the region $O$, equipped with the causal structures inherited from $\Mb$, 
and regarded as a spacetime (i.e., an object of $\Loc$) in its own right, and $\iota_O$ is the
inclusion map, which becomes a morphism $\iota_O:\Mb|_O\to\Mb$ (see Fig.~\ref{fig:kinematic}).
The kinematic net is automatically isotonous, owing to the functorial nature of $\Af$: if
$O_1\subset O_2$ then $\iota_{O_1}$ factors through $\iota_{O_2}$ and hence $\Af^\kin(\Mb;O_1)\subset \Af^\kin(\Mb;O_2)$.
The other assumptions we make can now be stated. 
\begin{description}
\item[\em Einstein Causality] If $O_i\subset \Mb$ are spacelike
separated then
\begin{equation}
[\Af^\kin(\Mb;O_1),\Af^\kin(\Mb;O_2)] = 0.
\end{equation}
\item[\em Timeslice] $\Af$ maps Cauchy morphisms to $\CAlg$-isomorphisms.
\end{description}
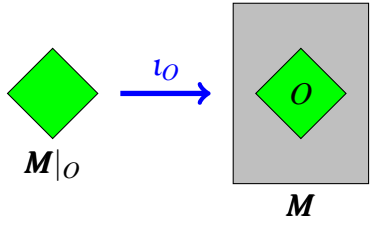
\begin{figure}[t]
\begin{center}
\begin{tikzpicture}[scale=0.6]
\draw[fill=lightgray] (-7,0) ++(-1.5,0) -- ++(3,0) -- ++(0,4) --
++(-3,0) -- cycle;
\draw[fill=Green] (-7,2) +(-1,0) -- +(0,-1) -- +(1,0) -- +(0,1) -- cycle;
\node[anchor=north] at (-7,0) {$\Mb$};
\draw[fill=Green] (-12.5,2) +(-1,0) -- +(0,-1) -- +(1,0) -- +(0,1) -- cycle;
\node at (-7,2) {$O$};
\node[anchor=north] at (-12.5,1) {$\Mb|_O$};
\draw[color=blue,line width=2pt,->] (-11,2) -- (-9,2) node[pos=0.5,above] {$\iota_O$};
\end{tikzpicture}
\end{center}\caption{Illustration of the objects involved in defining the kinematic net}
\label{fig:kinematic}
\end{figure}

\paragraph{The split property}
The algebras produced by $\Af$ are abstract $C^*$-algebras. However, any state $\omega$ on $\Af(\Mb)$ induces a corresponding GNS representation $(\HH,\pi,\Omega)$ such that
$\omega(A)=\ip{\Omega}{\pi(A)\Omega}$ for all $A\in\Af(\Mb)$. Passing to the representation
and taking weak closures, we obtain a net of von Neumann algebras 
$O\mapsto\pi(\Af^\kin(\Mb;O))''$ associated to $\Af$, $\Mb$ and $\omega$. This is the setting in which the split property may be defined.
\begin{definition} 
Let $\Af:\Loc\to\CAlg$ be a locally covariant QFT and $\Mb\in\Loc$. 
A state $\omega$ on $\Af(\Mb)$ has the \emph{split property} for
a regular Cauchy pair $(S,T)$ if, in its GNS representation ($\HH,\pi,\Omega)$, there is a type-$I$ factor $\Ngth$ such that
\begin{equation}
\Rgth_S \subset \Ngth \subset  \Rgth_T ,
\end{equation}
where $\Rgth_U = \pi(\Af^\kin(\Mb;D_\Mb(U)))''$ for $U=S,T$.
\end{definition}
Apart from the use of regions based on regular Cauchy pairs, this is essentially the same formulation
as used in Minkowski space. The present approach allows us to establish a key inheritance property \cite[Remark 3.2]{Few_split:2015}. Suppose that 
$\omega$ is split for $(S,T)$. Then if $(\tilde{S},\tilde{T})$ is some other Cauchy pair in $\Mb$ with $(S,T)\prec_\Mb (\tilde{S},\tilde{T})$, we have $\tilde{S}\subset D_{\Mb}(S)$, $T\subset D_{\Mb}(\tilde{T})$, so also
\begin{equation}
D_\Mb(\tilde{S})\subset D_\Mb(S), \qquad 
D_\Mb(\tilde{T})\subset D_\Mb(T).
\end{equation}
Hence by isotony, $\Rgth_{\tilde{S}}  \subset  \Rgth_S \subset \Ngth \subset  \Rgth_T  \subset 
 \Rgth_{\tilde{T}}$, so $\omega$ is split for $(\tilde{S},\tilde{T})$.
Moreover, the split property is also stable under Cauchy morphisms:
\begin{lemma}
Suppose $\psi:\Mb\to\Nb$ is a Cauchy morphism and let 
$\omega_\Mb$, $\omega_\Nb$ be a states on $\Af(\Mb)$, $\Af(\Nb)$. 
Then $\omega_\Nb$ is split for $(\psi(S),\psi(T))$ if and only if $\Af(\psi)^*\omega_\Nb$ is
split for $(S,T)$.  Consequently, as $\Af(\psi)$ is an isomorphism, $\omega_\Mb$ is 
split for $(S,T)$ if and only if $(\Af(\psi)^{-1})^*\omega_\Mb$
is split for $(\psi(S),\psi(T))$.
\end{lemma}
\begin{proof} $\omega_\Nb$ and  $\Af(\psi)^*\omega_\Nb$ have unitarily equivalent GNS representations and the split inclusion is preserved because the type-I property and factor properties are stable. 
\end{proof} 

The main result is:
\begin{theorem}[Rigidity of the split property~{\cite[Thm 3.4]{Few_split:2015}}]
\label{thm:rigid_split}
Suppose that $\Af$ is a locally covariant QFT, $\Mb,\Nb\in\Loc$ have oriented-diffeomorphic Cauchy surfaces, and  $\omega_\Nb$ is a state on $\Af(\Nb)$ that is split for all regular Cauchy pairs in $\Nb$.  
 
Given any regular Cauchy pair $(S_\Mb,T_\Mb)$ in $\Mb$, there is a chain of Cauchy morphisms
between $\Mb$ and $\Nb$ inducing an isomorphism $\nu:\Af(\Mb)\to\Af(\Nb)$ 
such that $\nu^*\omega_\Nb$ has the split property for $(S_\Mb,T_\Mb)$.  
\end{theorem}
\begin{proof}[Outline proof]
First, because $\Mb$ and $\Nb$ have Cauchy surfaces
related by an orientation-preserving diffeomorphism, we may assume without loss that
both spacetimes have common underlying manifold $\RR\times\Sigma$, with the tangent
vector to the first factor being future-directed timelike, and so that $\Sigma_\Mb=\{\tau_\Mb\}\times \Sigma$
for some $\tau_\Mb\in\RR$. The next step is to construct a globally hyperbolic metric on $\RR\times\Sigma$ that interpolates between that of $\Mb$, with which it agrees on $(-\infty,\tau_1)\times\Sigma$ and that of $\Nb$, with which it agrees on 
$(\tau_2,\infty)\times\Sigma$, where $\tau_\Mb<\tau_1<\tau_2$. This also has to be arranged
so that for some $\tau_*\in(\tau_1,\tau_2)$ and $\tau_\Nb >\tau_2$ there are regular Cauchy pairs
$(S_*,T_*)$ and $(S_\Nb, T_\Nb)$ lying in $\Sigma_*=\{\tau_*\}\times\Sigma$ and $\Sigma_\Nb=\{\tau_\Nb\}\times\Sigma$ respectively, so that
\begin{equation}
(S_\Nb,T_\Nb) \prec_\Ib (S_*,T_*)\prec_\Ib (S_\Mb,T_\Mb),
\end{equation}
where $\Ib$ is the spacetime defined by the interpolating metric. 
Defining $\Pb=\Mb|_{(-\infty,\tau_1)\times\Sigma}$ and $\Fb=\Nb|_{(\tau_2,\infty)\times\Sigma}$,
it is clear that the obvious inclusions induce a chain of Cauchy mophisms as illustrated in Fig.~\ref{fig:chain}. 
 
The inheritance and timeslice properties now allows us to pass the split property along the chain,
as follows: $\omega_\Nb$ is split for $(S_\Nb,T_\Nb)$,  so therefore
\begin{tightitemize}
\item by timeslice $\omega_\Fb=\Af(\delta)^*\omega_\Nb$ is split for $(S_\Nb,T_\Nb)$ (regarded as a regular Cauchy pair in $\Fb$) 
\item by timeslice $\omega_\Ib =( \Af(\gamma)^{-1})^*\omega_\Fb$ is split for $(S_\Nb,T_\Nb)$ in $\Ib$
\item by inheritance, $\omega_\Ib$ is split for $(S_*,T_*)$ and hence $(S_\Mb,T_\Mb)$ in $\Ib$
\item  by timeslice $\omega_\Pb=\Af(\beta)^*\omega_\Ib$ is split for $(S_\Mb,T_\Mb)$ in $\Pb$
\item by timeslice $\omega_\Mb=(\Af(\alpha)^{-1})^*\omega_\Pb$ is split for $(S_\Mb,T_\Mb)$ in $\Mb$,
\end{tightitemize}
and the theorem is proved with $\nu=\Af(\delta)\Af(\gamma)^{-1}\Af(\beta)\Af(\alpha)^{-1}$. 
\end{proof}
\begin{figure}
\begin{center}
\newcommand{\fcol}{Gold}
\newcommand{\bcol}{blue}
\begin{tikzpicture}[scale=0.8]
\draw[fill=\bcol] (-6,0) -- ++(2,0) -- ++(0,4) -- ++(-2,0) -- cycle; 
\draw[fill=\fcol] (6,0) -- ++(2,0) -- ++(0,4) -- ++(-2,0) -- cycle; 
\draw[fill=\bcol] (0,0) -- ++(2,0) -- ++(0,4) -- ++(-2,0) -- cycle; 
\draw[top color = \fcol,bottom color=\bcol,color=\bcol] (0,1) -- ++(2,0) -- ++(0,2) -- ++(-2,0) -- cycle;
\draw[fill=\fcol,color=\fcol] (0,3) -- ++(2,0) -- ++(0,1) -- ++(-2,0) -- cycle;
\draw (0,0) -- ++(2,0) -- ++(0,4) -- ++(-2,0) -- cycle;

\draw[fill=\fcol] (3,3) -- ++(2,0) -- ++(0,1) -- ++(-2,0) -- cycle;
\draw[fill=\bcol] (-3,0) -- ++(2,0) -- ++(0,1) -- ++(-2,0) -- cycle;
 
\node[anchor=north] at (7,-0.25) {$\Nb$};
\node[anchor=north] at (4,-0.25) {$\Fb$};
\node[anchor=north] at (1,-0.25) {$\Ib$};
\node[anchor=north] at (-2,-0.25) {$\Pb$};
\node[anchor=north] at (-5,-0.25) {$\Mb$};
\draw[line width=1pt,->] (4.25,-0.7)--(6.75,-0.7)  node[pos=0.4,below]{$\delta$};
\draw[line width=1pt,->] (3.75,-0.7)--(1.25,-0.7) node[pos=0.4,below]{$\gamma$};
\draw[line width=1pt,->] (-1.75,-0.7)--(0.75,-0.7) node[pos=0.4,below]{$\beta$};
\draw[line width=1pt,->] (-2.25,-0.7)--(-4.7,-0.7) node[pos=0.4,below]{$\alpha$};

\draw[line width=1pt,color=red] (6,3.5) -- +(2,0);
\draw[line width=1pt,color=red] (3,3.5) -- +(2,0);
\draw[line width=1pt,color=red] (0,3.5) -- +(2,0);
\draw[line width=1pt,color=red] (6,2) -- +(2,0);
\draw[line width=1pt,color=red] (0,2) -- +(2,0);
\draw[line width=1pt,color=red] (-6,2) -- +(2,0);
\draw[line width=1pt,color=red] (0,0.5) -- +(2,0);
\draw[line width=1pt,color=red] (-3,0.5) -- +(2,0);
\draw[line width=1pt,color=red] (-6,0.5) -- +(2,0);

\node[anchor=east] at (-6,0.5) {$\Sigma_\Mb$};
\node[anchor=east] at (0,2) {$\Sigma_*$};
\node[anchor=west] at (8,3.5) {$\Sigma_\Nb$};
\end{tikzpicture}
\end{center}
\caption{The chain of Cauchy morphisms in the proof of Theorem~\ref{thm:rigid_split}.}
\label{fig:chain}
\end{figure}
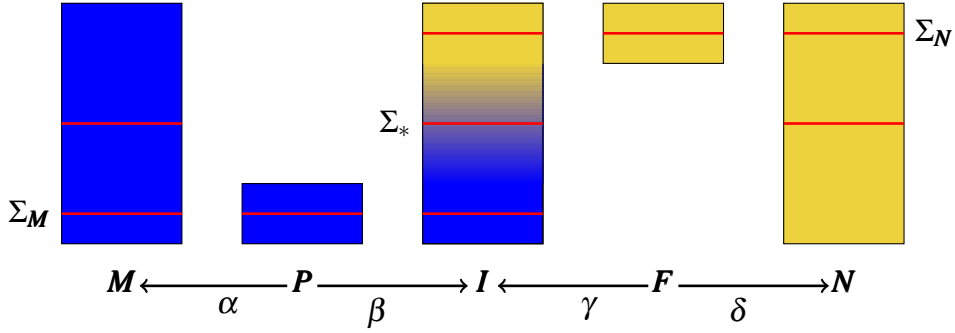

Now consider an arbitrary globally hyperbolic spacetime $\Mb\in\Loc$. Choose any smooth spacelike Cauchy surface $\Sigma$ of $\Mb$ with an induced orientation $\ogth$ (so that $\tgth\wedge\ogth$ is positively
oriented for any future-pointed timelike one-form $\tgth$), and endow $\Sigma$ with a complete Riemannian metric $h$~\cite{NomizuOzeki1961}. Then one obtains an ultrastatic spacetime $\Nb$
with metric~\eqref{eq:ultra} that is also an object of $\Loc$ and clearly has smooth spacelike Cauchy surfaces oriented-diffeomorphic to those of $\Mb$. Together with Theorem~\ref{thm:rigid_split}
this shows that the problem of establishing a split property in $\Mb$ [in the sense that, for any regular Cauchy pair, there exists a state for which the inclusion is split] can be reduced to the 
ultrastatic case, for which nuclearity criteria can be employed (and have been proved in certain models~\cite{Verch_nucspldua:1993,DAnHol:2006}).
 
Theorem~\ref{thm:rigid_split} can be extended in various ways. For example, one may construct states
that are split for finitely many regular Cauchy pairs lying in a common Cauchy surface. Again, suppose that
the theory has a state space $\Sf$ obeying the timeslice property; that is, an assignment to each $\Mb\in\Loc$ of a subset $\Sf(\Mb)$ of states on $\Af(\Mb)$ subject to the contravariance requirement 
\begin{equation}\label{eq:Sf}
\Af(\psi)^*\Sf(\Nb)\subset \Sf(\Nb)
\end{equation}
for each morphism $\psi:\Mb\to\Nb$, with equality if $\psi$ is Cauchy,
and also the requirement that each $\Sf(\Mb)$ be
closed under convex combinations and operations induced by $\Af(\Mb)$.  Suppose further that $\Sf(\Nb)$ contains a state $\omega_\Nb$ that is split for all regular Cauchy pairs. Then the state $\omega_\Mb$ constructed by the theorem 
belongs to $\Sf(\Mb)$ because $\nu$ is formed from Cauchy morphisms. 
Furthermore, if $\Sf$ obeys a condition of local quasiequivalence, one may deduce
in addition that every state in $\Sf(\Mb)$ is split for all regular Cauchy pairs in $\Mb$ 
and also for any pair of regions $O_i\in\OO(\Mb)$ such that
\begin{equation}
O_1\subset D_\Mb(S) \qquad D_\Mb(T)\subset O_2
\end{equation}
for a regular Cauchy pair $(S,T)$. Details appear in \cite{Few_split:2015}.

\paragraph{(Partial) Reeh--Schlieder property} It should be clear that the
proof of Theorem~\ref{thm:rigid_split} depends on only a few ideas:
the formulation in terms of regular Cauchy pairs, and the inheritance and
timeslice properties. Similar arguments apply to other properties of 
quantum field theory (a key theme of~\cite{FewVerch_aqftincst:2015} is the 
utility of this and related rigidity arguments). In particular, one obtains
a streamlined version of Sanders' results on Reeh--Schlieder properties~\cite{Sanders_ReehSchlieder}. 
\begin{definition}
Let $\Af:\Loc\to\CAlg$ be a locally covariant QFT and $\Mb\in\Loc$. A state $\omega$ on $\Af(\Mb)$ has the \emph{Reeh--Schlieder property} for a regular Cauchy pair $(S,T)$ if, in the GNS representation 
($\HH,\pi,\Omega)$, $\Omega$ is cyclic for $\Rgth_S$ and separating for $\Rgth_T$. 
\end{definition}
That is, vectors of the form $A\Omega$ ($A\in\Rgth_S$) are dense in $\HH$, and no $B\in\Rgth_T$
can annihilate $\Omega$. Note that the state is not assumed to be cyclic for all local algebras, 
so one might refer to this as a partial Reeh--Schlieder property, by comparison with the original
Minkowski-space result~\cite{ReehSchlieder:1961}. Our definition obeys an inheritance rule like that of the split property, but with a reversed ordering: if $\omega$ is Reeh--Schlieder for $(S,T)$ in $\Mb$, then it is also Reeh--Schlieder for every $(\tilde{S},\tilde{T})$ with 
\begin{equation}
(\tilde{S},\tilde{T})\prec_\Mb (S,T)
\end{equation}
simply because the separating property is inherited by subalgebras and cyclicity by algebras containing
the given one. One can also prove that our definition is stable with respect to the timeslice condition,
and this easily establishes:
\begin{theorem}[Rigidity of Reeh--Schlieder~{\cite[Thm 3.11]{Few_split:2015}}, cf. \cite{Sanders_ReehSchlieder}]  
Suppose that $\Af$ is a locally covariant QFT, $\Mb,\Nb\in\Loc$ have oriented-diffeomorphic Cauchy surfaces, and  $\omega_\Nb$ is a state on $\Af(\Nb)$ that is Reeh--Schlieder for all regular Cauchy pairs.  
 
Given any regular Cauchy pair $(S_\Mb,T_\Mb)$ in $\Mb$, there is a chain of Cauchy morphisms
between $\Mb$ and $\Nb$ inducing an isomorphism $\nu:\Af(\Mb)\to\Af(\Nb)$ such
that $\nu^*\omega_\Nb$ has the Reeh--Schlieder property for $(S_\Mb,T_\Mb)$.
\end{theorem}
\begin{proof} Invert the argument for the split property,
replacing $\prec$ by $\succ$. \end{proof}

As with the split property, the above Reeh--Schlieder property in general spacetimes
can now be traced to the ultrastatic case. Here, one can give a \emph{tube condition} 
(satisfied in general Wightman theories in Minkowski space~\cite{Borchers:1961} and in particular linear  models in curved spacetimes~\cite{Stroh:2000})  that implies that ground states obey the full Reeh--Schlieder property.  Thus the (partial) Reeh--Schlieder property becomes a reasonable expectation
in general locally covariant theories. Our result admits various extensions, of which the most
important is that the proofs of the split and Reeh--Schlieder results may be combined
to yield states that are \emph{both Reeh--Schlieder and split}
for (finitely many) regular Cauchy pairs in a common Cauchy surface. 
Consequently, given any regular Cauchy pair $(S,T)$ one may find a state $\omega$ with
GNS vector $\Omega$ so that $(\Rgth_S,\Rgth_T,\Omega)$ is
a standard split inclusion, leading to applications analogous to those described in section~\ref{sec:split}.  Finally, let us mention that
a slightly different definition of the split property has been studied
in~\cite{BrFrImRe:2014} and used to discuss the tensorial structure
of locally covariant theories.

\paragraph{Distal split}  Suppose we have a locally covariant theory $\Af$ with state space $\Sf$. In $d$-dimensional Minkowski spacetime $\Mb$, we define the \emph{splitting distance} of $S\subset\RR^{d-1}$ by
\begin{equation}
d(S) = \inf\{\rho>0: \text{$(S, B(S,\rho))$ is split for some state in $\Sf(\Mb)$}\},
\end{equation}
where $B(S,\rho)$ is the open Euclidean ball of radius $\rho$ about $S$. 
If $d(S)$ is finite but nonzero we say that the \emph{distal split property} holds. It is of interest to understand what happens to such models in the locally covariant setting, 
assuming that the state space obeys local quasiequivalence.
Suppressing some $\epsilon$'s (see \cite[\S 3.4]{Few_split:2015} for the precise statements and proof) one
has the following. If $f\in\text{Diff}(\RR^{d-1})$ with uniformly bounded derivatives, and $r>d(f(S))$, then 
\begin{equation}\label{eq:splitbd}
d(S) \le \inf\{\rho>0: B(f(S),r)\subset f(B(S,\rho))\}
\end{equation}
and hence
\begin{equation}
d(S) \le \kappa d(f(S)),
\end{equation}
where $\kappa$ is the supremum of the norm of the derivative $\|D(f^{-1})\|$ over $B(f(S),r)\setminus f(S)$. 

As a first example, taking $f(\xb)=\xb/\lambda$ implies that $d(\lambda S)\le \lambda d(S)$
and hence the existence of a uniform splitting distance ($d(S)=d_0\ge 0$ for all $S$) actually 
implies the split property ($d_0=0$). Of course, what this means is that models such as
that of~\cite[Thm 4.3]{DAnDopFreLon:1987} given in \eqref{eq:bad_mn} cannot be compatible with all the hypotheses
required -- one might well suspect that local quasiequivalence and/or the timeslice property
become problematic. 

For a second example, let $S$ be an open ball and design $f\in\text{Diff}(\RR^{d-1})$ with the properties that $f(S)=S$ and $B(S,d(S))\subset f(B(S,\frac{1}{2}d(S)))$, and so that
$f$ acts trivially outside a compact set.
Then (again suppressing $\epsilon$'s) one obtains
\begin{equation}
d(S)\le \frac{1}{2}d(S), 
\end{equation}
and hence the only possible splitting distances are $0$ or $\infty$. 
This proves a stronger result that the distal split property implies the split property (at least for sets diffeomorphic to balls), given our other hypotheses.
 
The proof of~\eqref{eq:splitbd} takes its inspiration from cosmological inflation (see Fig.~\ref {fig:inflate}). Let $S_1$ be an open bounded subset in a constant time hypersurface and let $T_1=B(S_1,r)$ where  $r>d(S_1)$, so that some state is split for $(S_2,T_2)$. Then we may infer that the state is split for suitable regular Cauchy pairs lying in a constant time hypersurface to the past. In the left-hand half of Fig.~\ref{fig:inflate} we illustrate the situation in Minkowski space: $S_2$ is has a smaller diameter than $S_1$ while $T_2$ has a larger diameter than $T_1$, so this argument
(not very helpfully) gives an upper bound on $d(S_2)$ that is larger than $d(S_1)$. In the right-half of the figure, however, we imagine that
the metric undergoes a period of inflation between the regions
around the two hypersurfaces and outside the Minkowski $D(S_1)$).
If we arrange matters so that $T_2$ fits within a ball of radius $\kappa r\diam(S_2)/\diam(S_1)$
for $\kappa<1$ then we obtain a tighter bound on $d(S_2)$ than $r$ is on $d(S_1)$.  
The $\epsilon$'s suppressed in our discussion relate to the spreading of lightcones (at the
Minkowski speed of light) which can be made arbitrarily small by considering nearby hypersurfaces.

Actually, there is some physics to go alongside what would otherwise
appear to be a geometric trick. As described earlier, nonzero splitting distances are 
associated with the existence of a critical temperature, above 
which the thermal equilibrium states are no longer locally quasiequivalent to the vacuum. On the other hand, inflation cools temperature (to the future); conversely, one could expect that some states of subcritical temperature near the later hypersurface have supercritical temperatures
at the earlier one. This leads to a contradiction between the assumptions
of the timeslice axiom and local quasiequivalence.

\begin{figure}\tdplotsetmaincoords{75}{90}
\pgfmathsetmacro{\rvec}{.8}
\pgfmathsetmacro{\thetavec}{15}
\pgfmathsetmacro{\phivec}{60}
\begin{center}
\begin{tikzpicture}[scale=4,tdplot_main_coords]
\coordinate (O) at (0,0,0);
\coordinate (Q) at (0,0,-0.4);

\tdplotdrawarc[color=lightgray,fill=lightgray]{(Q)}{0.8}{0}{360}{anchor=north}{}
\tdplotdrawarc[color=vlg,fill=vlg]{(O)}{0.8}{0}{360}{anchor=north}{}
\tdplotdrawarc{(O)}{0.35}{0}{360}{anchor=north}{}
\tdplotdrawarc{(O)}{0.6}{0}{360}{anchor=north}{}
\tdplotdrawarc{(Q)}{0.25}{0}{360}{anchor=north}{}
\draw[dotted] (0,0.35,0) -- (0,0.25,-0.4);
\draw[dotted] (0,-0.35,0) -- (0,-0.25,-0.4);
\node at (0,0.0,0.0) {$S_1$};
\node at (0,0.475,0) {$T_1$};
\node at (0,0.0,-0.4) {$S_2$};
\draw[<->] (-0.247,0.247,0) -- (-0.424,0.424,0);
\node at (-0.3,0.3,0.1) {$r$};

\draw[dashed] (0,0.6,0) -- (0,0.7,-0.4);
\draw[dashed] (0,-0.6,0) -- (0,-0.7,-0.4);
\tdplotdrawarc{(Q)}{0.7}{0}{360}{anchor=north}{}
\node at (0,0.55,-0.4) {$T_2$};
\end{tikzpicture}\hfil
\begin{tikzpicture}[scale=4,tdplot_main_coords]
\coordinate (O) at (0,0,0);
\coordinate (Q) at (0,0,-0.4);

\tdplotdrawarc[color=lightgray,fill=lightgray]{(Q)}{0.8}{0}{360}{anchor=north}{}
\tdplotdrawarc[color=vlg,fill=vlg]{(O)}{0.8}{0}{360}{anchor=north}{}
\tdplotdrawarc{(O)}{0.35}{0}{360}{anchor=north}{}
\tdplotdrawarc{(O)}{0.6}{0}{360}{anchor=north}{}
\tdplotdrawarc{(Q)}{0.25}{0}{360}{anchor=north}{}
\draw[dotted] (0,0.35,0) -- (0,0.25,-0.4);
\draw[dotted] (0,-0.35,0) -- (0,-0.25,-0.4);
\node at (0,0.0,0.0) {$S_1$};
\node at (0,0.475,0) {$T_1$};
\node at (0,0.0,-0.4) {$S_2$};
\draw[<->] (-0.247,0.247,0) -- (-0.424,0.424,0);
\node at (-0.3,0.3,0.1) {$r$};

\draw[dashed]  (0,0.6,0) .. controls (0,0.65,-0.2) and  (0,0.6,-0.17) .. (0,0.5,-0.2);
\draw[dashed]  (0,0.5,-0.2) .. controls (0,0.4,-0.23) and  (0,0.35,-0.2) .. (0,0.4,-0.4);
\draw[dashed]  (0,-0.6,0) .. controls (0,-0.65,-0.2) and  (0,-0.6,-0.17) .. (0,-0.5,-0.2);
\draw[dashed]  (0,-0.5,-0.2) .. controls (0,-0.4,-0.23) and  (0,-0.35,-0.2) .. (0,-0.4,-0.4);
\tdplotdrawarc{(Q)}{0.4}{0}{360}{anchor=north}{}
\node at (0,0.31,-0.4) {$T_2'$};
\end{tikzpicture}
\end{center}
\caption{Illustration of the `inflationary' proof of the distal split bound.}
\label{fig:inflate}
\end{figure}
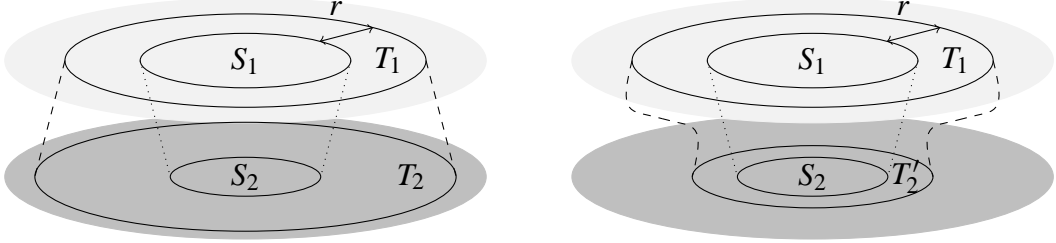

\section{Concluding remarks}

It is hoped that this paper has explained the physical significance of the 
split and nuclearity properties, and also explained -- in the context of the split and Reeh--Schlieder properties -- a general proof
strategy for establishing structural properties of locally covariant quantum field theories. As a further example, we mention that Lechner and Sanders have recently applied the machinery of regular Cauchy pairs in their proof of modular nuclearity~\cite{LecSan:2015}. We have sketched
links with the theory of QEIs, which require further investigation. At first sight, the split and nuclearity properties might not seem to have much
in common with each other, or with the QEIs. But in fact, all three can
be viewed as expressions of the uncertainty principle: this is most
obvious in the case of the QEIs, which provide rigorous time-energy
inequalities, but it is also true of the split property (the need for a collar
region to guarantee independence indicating a loss of sharp localisation) and nuclearity (understood as a
constraint on the number of states per unit volume in phase space). 
They therefore occupy a fascinating position at the nexus of relativity and quantum theory. 
 
\paragraph{Acknowledgement} Section~\ref{sec:cst} of this paper is based on \cite{Few_split:2015} and a talk given at the annual meeting of the Deutschen Mathematiker-Vereinigung (DMV) in Hamburg, September 2015. I thank
the organisers of the Mini-Symposium {\em Algebraic Quantum Field Theory on Lorentzian Manifolds} for the invitation to speak and
to produce this paper. I am also grateful to Martin Porrmann for
discussions on nuclearity indices in the early years of the present millennium.

 \appendix

\section{Proof of Theorem~\ref{thm:mbound}}
\label{sect:constr}
 
Recall that $m_0\ge 0$ has been fixed and that $f\in\CoinX{\RR}$ is nonnegative, even, with unit integral, and has a Fourier transform that is real, even, nonnegative and bounded from below by  
\begin{equation}
\hat{f}(u)\ge \varphi(|u|)
\end{equation}
on $[m_0,\infty)$, where $\varphi:[m_0,\infty)\to\RR^+$ is monotone decreasing. Note that nonnegativity of $f$ and $\hat{f}$ implies that $\hat{f}(m_0)\le \hat{f}(0)=1$ and hence $\varphi(u)\le 1$ for all $u\ge m_0$.

Consider a single Klein--Gordon field of mass $m\ge m_0$ on the symmetric Fock space $\FF$ over $L^2(\RR^3,d^3\kb/(2\pi)^3)$, and define the usual annihilation operators $a(\kb)$ by 
\begin{equation}
(a(\kb)\Psi)^{(n)}(\kb_1,\ldots,\kb_n) = \sqrt{n+1}\Psi^{(n+1)}(\kb,\kb_1,\ldots,\kb_n),
\end{equation}
where for $\Psi\in\FF$, $\Psi^{(n)}$ denotes its $n$-particle component.\footnote{The annihilation operators used in Section~\ref{sec:nuc} are $\sa(u)=\int d^3\kb/(2\pi)^3  a(\kb) \overline{u(\kb)}$.}
Writing $a^\dagger(\kb)$ for the adjoint of $a(\kb)$ as a quadratic form,
the canonical commutation relations are
\begin{equation}
[a(\kb),a(\kb')] = 0, \qquad
[a(\kb),a^\dagger(\kb')] = (2\pi)^3 \delta(\kb-\kb')\II\,,
\end{equation} 
and the quantum field is given by
\begin{equation}
\Phi(x) = \int\frac{d^3\kb}{(2\pi)^3\sqrt{2\omega}}
\left(a(\kb)e^{-ik_a x^a} + a^\dagger(\kb) e^{ik_a x^a}\right)\,,
\end{equation}
in which $k^a=(\omega,\kb)$ with $\omega=(\|\kb\|^2+m^2)^{1/2}$.
 
The energy density (with respect to the standard time coordinate) is a sum of Wick squares
\begin{equation}
\rho_m(x) = \frac{1}{2}\left({:}(\nabla_0\Phi(x))^2{:} + \sum_{i=1}^3 {:}(\nabla_i\Phi(x))^2{:} + m^2{:}\Phi(x)^2{:}\right),
\end{equation}
so, again in a quadratic form sense, 
\begin{align}
\rho_m(x) &= \int\frac{d^3\kb}{(2\pi)^3}\frac{d^3\kb'}{(2\pi)^3}\frac{1}{4\sqrt{\omega\omega'}} \left\{(\omega\omega'+\kb\cdot\kb'+m^2)e^{i(k-k')\cdot
x}a^\dagger(\kb)a(\kb') \right. \nonumber \\ 
&\qquad\qquad\qquad \left. - (\omega\omega'+\kb\cdot\kb'-m^2)e^{-i(k+k')\cdot
x}a(\kb)a(\kb')\right\} + \text{H.C.},
\end{align}
where $\text{H.C.}$ denotes the hermitian conjugate.

Next, choose any smooth, symmetric, nonnegative function
$B:\RR^3\times\RR^3\to\RR$, with compact support obeying
\begin{equation}
\supp B\subset\{(\ub,\ub'): \|\ub\|,\|\ub'\|\in[{\textstyle\frac{1}{2}},1];\quad |\theta(\ub,\ub')|<\pi/3\},
\label{eq:Bsupp}
\end{equation}
where $\theta(\ub,\ub')$ is the angle between the vectors $\ub$, $\ub'$, and normalised so that
\begin{equation}\label{eq:Bnorm}
\int \frac{d^3\ub}{(2\pi)^3}\, \frac{d^3\ub'}{(2\pi)^3} B(\ub,\ub') =1\,.
\end{equation}
The function $C:\RR^3\times\RR^3\to\RR$
\begin{equation}
C(\ub,\ub') = \int \frac{d^3\ub''}{(2\pi)^3} B(\ub'',\ub)B(\ub'',\ub')
\end{equation}
is then pointwise nonnegative with support obeying 
\begin{equation}
\supp C\subset\{(\ub,\ub'): \|\ub\|,\|\ub'\|\in[{\textstyle\frac{1}{2}},1] \}
\label{eq:Csupp}.
\end{equation}

We now define, for $\lambda>0$, the vacuum-plus-two-particle 
superposition 
\begin{equation}
\Psi_{m,\lambda} = \sN_{m,\lambda}\left[
\Omega + \frac{\lambda}{\sqrt{2}}\int \frac{d^3\kb}{(2\pi)^3}\,\frac{d^3\kb'}{(2\pi)^3}
b(\kb,\kb')a^\dagger(\kb)a^\dagger(\kb')\Omega\right]\,,
\end{equation}
where $\Omega\in\FF$ is the Fock vacuum vector, $\sN_{m,\tau,\lambda}$ is a normalisation constant and
\begin{equation}\label{eq:bdef}
b(\kb,\kb') = \frac{\varphi(2\sqrt{2} m)}{m^3}B(\kb/m,\kb'/m)\,.
\end{equation}
That is, $\Psi_{m,\lambda}^{(0)}=\sN_{m,\lambda}$, $\Psi_{m,\lambda}^{(2)}(\kb,\kb') =\sN_{m,\lambda} \lambda b(\kb,\kb')$ and all other components of $\Psi_{m,\lambda}$ vanish. 
As $b$ is compactly supported, each $\Psi_{m,\lambda}$ is
a Hadamard state. The normalisation constant is 
\begin{equation}
\sN_{m,\lambda} =
\left( 1+\lambda^2\varphi(2\sqrt{2} m)^2 \Tr C\right)^{-1/2} \ge \left( 1+\lambda^2  \Tr C\right)^{-1/2},
\end{equation}
where we have used $\varphi\le 1$ and employed the short-hand notation
\begin{equation}
\Tr C = \int\frac{d^3\ub}{(2\pi)^3} C(\ub,\ub)\,.
\end{equation}

Using the general formulae $\ip{\Omega}{a(\kb)a(\kb')\Psi} = \sqrt{2}\Psi^{(2)}(\kb,\kb')$ and
\begin{equation} 
\ip{\Psi}{a^\dagger(\kb)a(\kb')\Psi} = 2\int\frac{d^3\kb''}{(2\pi)^3}
\overline{\Psi^{(2)}(\kb'',\kb)}\Psi^{(2)}(\kb'',\kb'), 
\end{equation}
for vacuum-plus-two-particle superpositions $\Psi$, the expected normal ordered energy density is
\begin{align}
\ip{\Psi_{m,\lambda}}{\rho_m(x) \Psi_{m,\lambda}} &= |\sN_{m,\lambda}|^2\Re
\int \frac{d^3\kb}{(2\pi)^3}\,\frac{d^3\kb'}{(2\pi)^3}
\frac{1}{\sqrt{\omega\omega'}}\left(
\lambda^2 c(\kb,\kb')(\omega\omega'+\kb\cdot\kb'+m^2)e^{i(k-k')\cdot
x}\right.\nonumber\\
&\qquad\qquad\qquad\qquad
\left.
-\frac{\lambda}{\sqrt{2}} b(\kb,\kb')(\omega\omega'+\kb\cdot\kb'-m^2)e^{-i(k+k')\cdot x}
\right),
\end{align}
where
\begin{equation}\label{eq:cdef}
c(\kb,\kb') = 
\frac{\varphi(2\sqrt{2} m)^2 }{m^3}C(\kb/m,\kb'/m)\,.
\end{equation}
As $\Psi_{m,\lambda}$ is Hadamard, the
expectation value is smooth in $x$ and we can therefore
average against $f(t)$, using the fact
that $\hat{f}$ is real, to find
\begin{align}\label{eq:rho_calc}
\int
\ip{\Psi_{m,\lambda}}{\rho_m(t,\Ob) \Psi_{m,\lambda}}f(t)\,dt &=
|\sN_{m,\lambda}|^2\int \frac{d^3\kb}{(2\pi)^3}\,\frac{d^3\kb'}{(2\pi)^3}
\frac{1}{\sqrt{\omega\omega'}}\nonumber\\
&\qquad\times\left(
\lambda^2 c(\kb,\kb')(\omega\omega'+\kb\cdot\kb'+m^2)
\hat{f}(\omega'-\omega)\right.\nonumber\\
&\qquad
\left.
-\frac{\lambda}{\sqrt{2}} b(\kb,\kb')(\omega\omega'+\kb\cdot\kb'-m^2)\hat{f}(\omega+\omega')
\right).
\end{align}

We now seek an upper bound on this last quantity. First note that, 
for $(\kb,\kb')\in\supp c$, we have $\omega,\omega'\in [\sqrt{5}m/2,\sqrt{2}m]$ and
\begin{equation}
\frac{1}{\sqrt{\omega\omega'}}\left(\omega\omega'+\kb\cdot\kb'+m^2\right) \le
\frac{4m^2}{\sqrt{5}m/2}=
\frac{8m}{\sqrt{5}}\,,
\end{equation}
while if $(\kb,\kb')\in\supp b$ we have
\begin{equation}
\frac{1}{\sqrt{\omega\omega'}}\left(\omega\omega'+\kb\cdot\kb'-m^2\right)
\ge \frac{1}{m\sqrt{2}}
\left(\frac{m^2}{4}+\kb\cdot\kb' \right) \ge \frac{3m}{8\sqrt{2}}\,.
\end{equation}
Second, because $f$ and $\hat{f}$ are positive, $
\hat{f}(\omega'-\omega)\le
\hat{f}(0)=1$;  
furthermore, for $(\kb,\kb')\in\supp b$ (and so, in particular, 
$m_0<\sqrt{5}m_0\le \omega+\omega'\le 2\sqrt{2}m_0$) we have
\begin{equation}
\hat{f}(\omega+\omega')\ge \varphi(\omega+\omega')\ge \varphi(2\sqrt{2} m)\,.
\end{equation}
Accordingly, as $\lambda$ and the functions $b$ and $c$ are positive,
we obtain the bound
\begin{equation}
\text{L.H.S. of~\eqref{eq:rho_calc}}   \le 
|\sN_{m,\lambda}|^2\int \frac{d^3\kb}{(2\pi)^3}\,\frac{d^3\kb'}{(2\pi)^3}
\left(\lambda^2 c(\kb,\kb')\frac{8m}{\sqrt{5}} 
-\lambda b(\kb,\kb')\frac{3m}{16}\varphi(2\sqrt{2} m)
\right),
\end{equation}
the right-hand side of which can be written as
$-|\sN_{m,\lambda}|^2 P(\lambda) m^4 \varphi(2\sqrt{2} m)^2$,
where 
\begin{equation}
P(\lambda) = \frac{3}{16} \lambda -\lambda^2 \frac{8}{\sqrt{5}} \int 
\frac{d^3\ub}{(2\pi)^3}\,\frac{d^3\ub'}{(2\pi)^3}
C(\ub,\ub') ,
\end{equation}
as follows on inserting the definitions \eqref{eq:bdef}, \eqref{eq:cdef} of $b(\kb,\kb')$ and
$c(\kb,\kb')$ and using the normalisation \eqref{eq:Bnorm} of $B$. 
Now the quadratic $P(\lambda)$ has a positive maximum at some
$\lambda_0>0$ (note that $P$ and $\lambda_0$ are independent of $m$). Defining $\Psi_{m}=\Psi_{m,\lambda_0}$, we therefore obtain
\begin{equation}
\int
\ip{\Psi_{m}}{\rho_m(t,\Ob) \Psi_{m}}f(t)\,dt 
\le -\Gamma m^4 \varphi(2\sqrt{2} m)^2\,,
\end{equation}
where
\begin{equation}
\Gamma = \frac{P(\lambda_0)}{1+\lambda_0^2 \Tr C}
\end{equation}
depends only on the function $B$ (and not on $m$ or $\varphi$). This completes the proof of Theorem~\ref{thm:mbound}.

\end{document}